\documentclass[conference, compsoc]{IEEEtran}

\usepackage{amsfonts, amsmath, amsthm, amssymb}

\usepackage[shortlabels]{enumitem}
\usepackage{grffile}

\usepackage{tcolorbox}
\newtcolorbox{exbox1}{
leftupper=1mm
boxsep=0.0pt, left=1mm, right=1mm, top=0mm, bottom=1mm,
fontupper={\footnotesize\textbf{Example 1.}},
halign lower=center,  colback=white, arc=0pt, middle=0mm,
boxrule=0.5pt}

\newtcolorbox{exbox2}{
detach title, 
boxsep=0.5pt, left=1mm, right=1mm, top=1mm, bottom=1mm,
fontupper={\footnotesize\textbf{Example 2.}},
halign lower=center,  colback=white, arc=0pt,
boxrule=0.5pt}

\usepackage{booktabs}
\usepackage{multirow}
\usepackage{pifont}     
\usepackage{array}      
\usepackage{makecell} 

\usepackage{listings}
\usepackage{msc}
\usepackage{cite} 
\usepackage{amsmath,amssymb,amsfonts}
\usepackage{algorithmic}
\usepackage{graphicx}
\usepackage{textcomp}
\usepackage{xcolor}
\def\BibTeX{{\rm B\kern-.05em{\sc i\kern-.025em b}\kern-.08em
    T\kern-.1667em\lower.7ex\hbox{E}\kern-.125emX}}

\newcommand{\todo}[1]{{\color{red}{[TODO:}{#1]}}}

\newcommand{\ex}{{eq}}

\newcommand{\cl}{{cl}}
\newcommand{\EqUsr}{\mathcal{E}_{Usr}}
\newcommand{\EqDH}{\mathcal{E}_{DH}}
\newcommand{\Eqsimp}{\mathcal{E}_{simp}}

\newcommand{\Ind}{Ind}

\newcommand{\unify}{\operatorname{Unify}}
\newcommand{\Gen}{\operatorname{Gen}}
\newcommand{\NoCanc}{\operatorname{NoCanc}}
\newcommand{\toPoly}{\operatorname{toPoly}}
\newcommand{\dhInv}{\operatorname{inv}}
\newcommand{\dhExp}{\operatorname{exp}}
\newcommand{\roots}{\operatorname{rt}}

\newcommand{\rt}{\mathrm{rt}}

\newcommand{\dheq}{\circeq}

\newcommand{\R}{\mathcal{R}}
\newcommand{\E}{\mathcal{E}}
\newcommand\restrict[1]{\raisebox{-.5ex}{$|$}_{#1}}

\newcommand{\crule}[3]{\frac{\ #1\ }{\ #2\ }\ #3}

\newtheorem{definition}{Definition}
\newtheorem{theorem}{Theorem}
\newtheorem*{theorem*}{Theorem}

\newtheorem{lemma}{Lemma}

\newtheorem{defi/}{Definition}

    \newtheoremstyle{TheoremNum}
        {\topsep}{\topsep}              
        {\itshape}                      
        {}                              
        {\bfseries}                     
        {.}                             
        { }                             
        {\thmname{#1}\thmnote{ \bfseries #3}}
    \theoremstyle{TheoremNum}

\theoremstyle{definition}
\newtheorem{remark}{Remark}

\makeatletter
\def\thmhead@plain#1#2#3{%
  \thmname{#1}\thmnumber{\@ifnotempty{#1}{ }\@upn{#2}}%
  \thmnote{ {\the\thm@notefont#3}}}
\let\thmhead\thmhead@plain
\makeatother
\newtheorem*{rul}{Rule}
\makeatletter
\let\@old@begintheorem=\@begintheorem
\def\@begintheorem#1#2[#3]{%
  \gdef\@thm@name{#3}%
  \@old@begintheorem{#1}{#2}[#3]%
}
\def\namedthmlabel#1{\begingroup
   \edef\@currentlabel{\@thm@name}%
   \label{#1}\endgroup
}
\makeatother

\theoremstyle{definition}

\begin{document}

\title{Beyond the Finite Variant Property:\\
Extending Symbolic Diffie-Hellman Group Models}

\author{\IEEEauthorblockN{Sofia Giampietro}
\IEEEauthorblockA{\textit{ETH Zurich} \\
Zurich, Switzerland \\
sofia.giampietro@inf.ethz.ch}
\and
\IEEEauthorblockN{Ralf Sasse}
\IEEEauthorblockA{\textit{ETH Zurich} \\
Zurich, Switzerland \\
ralf.sasse@inf.ethz.ch}
\and
\IEEEauthorblockN{David Basin}
\IEEEauthorblockA{\textit{ETH Zurich} \\
Zurich, Switzerland \\
basin@inf.ethz.ch}
}
\maketitle

\begin{abstract}
  Diffie-Hellman groups are commonly used in cryptographic  protocols.  While most state-of-the-art, symbolic protocol verifiers support 
them to some degree, they do not support all mathematical operations possible in these groups. In particular, they lack support for exponent addition, as these tools reason about terms using unification, which is undecidable in the theory describing all Diffie-Hellman operators.

In this paper we approximate such a theory and propose a semi-decision procedure to determine whether a protocol, which may use \emph{all} operations in such groups, satisfies user-defined properties. We implement this approach by extending the Tamarin prover to support  the full Diffie-Hellman theory, including group element multiplication and hence addition of exponents. This is the first time a state-of-the-art tool can model and reason about such protocols.
We illustrate our approach's effectiveness with different case studies:  ElGamal encryption and MQV. 
Using Tamarin, we prove security properties of ElGamal, and we rediscover known attacks on MQV.  \end{abstract}

\section{Introduction}
\label{sec:intro}

Diffie-Hellman groups and their algebraic properties are the basis of many key exchange protocols, 
signature schemes, and symmetric encryption schemes. 
Diffie-Hellman groups are cyclic groups,  defined by a group operation $\cdot:G\times G \rightarrow G$ and a generator $g$.  For a natural number $n$, we define $g^n$ to denote $g\cdot ...\cdot g$ applied $n$ times.\footnote{For some groups (notably elliptic curves) it is common to denote the group operation with the addition sign $+:G\times G \rightarrow G$ and to denote the $n$-fold addition $g+..+g$ by $ng$. In this paper, we will use the \emph{multiplicative} notation and hence call this operation `multiplication'. Note that this operation is generally different from the usual multiplication (of natural numbers) between exponents. }
By the definition of the group operators, $(g^a)^b=g^{ab}$ and $g^a\cdot g^b=g^{a+b}$. Diffie-Hellman groups must furthermore satisfy the following property: finding $a$ from $g^a$ must be computationally hard. Examples of such groups are some modular arithmetic groups and some elliptic curves.   

Given the widespread use of Diffie-Hellman groups in security protocols, there has been considerable effort to support these group operations in automatic protocol verification tools. Most state-of-the-art tools (Tamarin, ProVerif, Maude-NPA, etc.) can model Diffie-Hellman groups in some way. However, their models are restricted: they support exponentiation and multiplication of exponents, but none can model the multiplication of group elements and hence the addition of exponents. 


Several challenges make modelling Diffie-Hellman group multiplication with these tools difficult. 
These tools operate in the symbolic model of cryptography, using equational theories to represent cryptographic primitives. Moreover, they rely on unification in their proof search. 
When the multiplication of group elements is considered, the exponents in Diffie-Hellman groups form a field. 
First, it is not straightforward to describe a field by an equational theory 
 and hence any approach taken will use an equational theory that \emph{approximates} fields.
Second, to be practically useful, such an approximation must include associativity and commutativity of both the addition and the multiplication operator, and distributivity of multiplication over addition.
Unfortunately, this interaction between distributivity and associativity makes unification both undecidable and infinitary \cite{diophantine}. Namely, the unification problem may have as a solution an infinite set of most general unifiers. So even if we had an algorithm enumerating all unifiers, this would be impractical for automated tools that use unification.  

In this paper, we provide a semi-decision procedure for determining whether a protocol model that uses the full Diffie-Hellman group structure satisfies a given property, for a class of user-defined properties. Our approach does \emph{not} rely on direct unification.  Our starting point takes as inspiration results from Dougherty and Guttman \cite{dougherty, Dougherty14Decidability}, who proposed a rewrite system containing both distributivity and associativity that closely approximates fields in the following sense: any theorem that can be proven using this rewrite system (henceforth called $\mathit{DH}$) holds in all Diffie-Hellman group structures. Conversely, any equation that holds in \emph{infinitely many} Diffie-Hellman structures is a theorem of the $\mathit{DH}$ rewrite system. 
Dougherty and Guttman showed how this system can be used to prove the security of a restricted class of protocols. In our work, we target a wider class of protocols, and show how to integrate this approach into the Tamarin prover \cite{Tamarin}. 
To the best of our knowledge, this is the first time an automated verification tool succeeds in analyzing protocols based on this mathematical structure. 

We restrict our results to protocols and properties that satisfy the following two conditions. 
\begin{description}
\item[C1] We assume that all terms that represent group elements belong to the same group. In other words, we assume there is a fixed group generator $g$, and all group elements are of the form $g^{e}$ for some exponent $e$ (possibly containing the addition operator).
\item[C2] Our approach is applicable to security properties that contain terms whose subterms do not cancel each other out. For example, this holds for any term that does not contain the $+$ or $\cdot$ operator. 
\end{description}

Since most protocols require parties to use a common fixed group, \textbf{C1} is a modest restriction. The assumption \textbf{C2} concerns the properties covered. In particular, it excludes making security verdicts about a term that can cancel itself out, i.e., which can be reduced to $1$ or $0$. It does not exclude the \emph{other} terms in the protocol canceling out. 
If a term can reduce to $1$ or $0$, then it is likely that it will not satisfy many basic security properties (e.g. secrecy). Hence we expect that in most cases, either we find an attack for the property, or the protocol's design and the property specification provide sufficient conditions to guarantee the non-cancellation assumption where necessary. 

The main idea of our approach is to first work with sub-terms in a simplified equational theory that does \emph{not} contain the group multiplication, which allows us to use Tamarin's current unification algorithm. 
In particular, to determine whether a target term is constructible, we first determine (via unification) whether its subterms are deducible. Since we assume they do not cancel each other out, if there is one such subterm that cannot be deduced, then we can already conclude that the target term cannot be constructed. If all subterms can be deduced, we then determine whether they can be combined, possibly using the group multiplication, to construct the target term. To do so, following \cite{Dougherty14Decidability}, we interpret the symbolic terms as algebraic objects and use algebraic methods, e.g., Gauss elimination. Effectively, we replace the unification problem for terms containing the group operation with a system of algebraic equations. 

We implement this approach by modifying Tamarin's constraint solving algorithm, and we show that the resulting constraint solving relation obtained is still sound and (under the non-cancellation property) also complete, in the sense that if Tamarin returns an attack with respect to a given property, then there is actually a protocol execution that violates the property and if Tamarin declares the protocol secure, then there is no execution violating the property. This holds within our symbolic model, i.e., our extension of Tamarin will find attacks that exploit particular properties of particular groups, rather only properties that can be described by the equational theory. 

To illustrate the effectiveness of our approach, we perform case studies on common protocols that use such Diffie-Hellman group operations, notably ElGamal and MQV. 

\smallskip
\textbf{Contributions.}
We see our work as making three contributions. First and foremost, we present a symbolic approach for reasoning about security protocols satisfying condition \textbf{C1} that use the full structure of Diffie-Hellman groups, and prove its soundness and completeness under the assumption \textbf{C2}. If the condition \textbf{C2} is not satisfied, Tamarin will still try to prove the protocol's security properties. If a property is disproven, then the attack is a valid attack, unconditionally. If a property is proven, Tamarin will mark which terms fail to be non-cancellable for the security proof to hold. Showing that they do not cancel out requires manual inspection. Alternatively, if the condition \textbf{C2} does \emph{not} hold, Tamarin has then shown that if an attack trace exists, it must leverage those particular terms canceling each other out.

Second, we implement this approach in the Tamarin prover, resulting in the first automated tool that can handle such an equational theory. This required adding new rules to its constraint solving algorithm and combining algebraic equation solving with Tamarin's proof search. 
 
Finally, we carry out case studies to show that this approach and its implementation in Tamarin is effective in practice. Namely, we prove ElGamal secure and we rediscover known attacks on MQV. Both proofs and attacks are found within minutes. These examples illustrate the feasibility, scalability, and efficiency of our approach when applied to complex protocols utilizing not only Diffie-Hellman, but also other primitives such as symmetric encryption.

An executable of our Tamarin extension, its source code, and the case studies are provided at \cite{zenodoDH}.

\section{Previous work}
\subsection{General reasoning for equational theories}
Prior work \cite{Alwen} extended ProVerif with an interface to an external first-order automated theorem prover to support properties modulo equational theories for which a complete unification algorithm is unavailable, by soundly translating these theories into universal axioms in first-order logic. However, these automated provers rely on resolution and will diverge if the axioms generate an infinite chain of inferences. The paper provides no results or experiments for our class of problems. Furthermore, the translation in \cite{Alwen} from the property to the automated theorem prover only covers secrecy properties. 

\subsection{Dougherty and Guttman's rewrite system}
State-of-the-art symbolic protocol verifiers represent messages as terms in a term algebra and operations as function symbols applied to these terms. The term algebra is equipped with an equational theory, representing the algebraic properties of the cryptographic primitives modeled. The equational theories are specified by sets of equations that describe how terms can be transformed or rewritten (by orienting equations left to right) according to these properties. To reason with such theories, the equations are often required to be convergent, namely that every term has a unique normal form (that cannot be reduced further using the equations) according to that theory.

We first introduce the rewrite theory we will use to represents fields, which is a slight modification of the theory presented by Dougherty and Guttman in \cite{dougherty}. We verify this theory is convergent 
modulo the associativity and commutativity of the theory's operators (abbreviated by AC) using the Aprove termination tool \cite{aprove} and the Maude Church-Rosser Checker \cite{church-rosser}.
We follow standard notions and conventions from term-rewriting, see \cite{Baader_Nipkow}. In particular $t:S$ indicates a term $t$ of sort $S$. 
\begin{definition}
\label{def:dheq}
Let $\Sigma_{DH}$ be an order-sorted signature with sorts $G$ and $E$, with the following function symbols and constants:
\begin{align*}
& \cdot : G\times G\rightarrow G & & e_G: \rightarrow G \\
&^{-1}: G\rightarrow G & & +,* : E\times E \rightarrow E \\
& 0: \rightarrow E  & & \dhExp: G\times E \rightarrow G \\
& - : E \rightarrow E & &  \mu : G \rightarrow \mathit{E}  \\
& \dhInv: \mathit{E} \rightarrow \mathit{E} & & 1: \rightarrow \mathit{E} 
\end{align*}
The operators obey the following equations:
\begin{itemize}
\item 
$(G, e_G, \cdot, ^{-1})$ is an abelian group.
\item 
$(E, +, 0, -, *, 1)$ is a commutative ring with identity.
\item 
$G$ is a right $E$-module for the operator $\dhExp$:
\begin{align*}
& (g^x)^y = g^{x*y} & &(g_1\cdot g_2)^x = g_1^x\cdot g_2^y & &   \\
& g^1 = g  & &g^{x+y} = g^x\cdot g^y  & & e_G^x = e_G
\end{align*}
\item For all terms $u$ and $v$ of sort $\mathit{E}$:
\begin{align*}
&u * \dhInv(u) = 1 & &  \dhInv(1) = 1  \\
& \dhInv(-u) = -\dhInv(u)  & & \dhInv(\dhInv(u))=u \\
& \dhInv(u*v) = \dhInv(u)*\dhInv(v) & & 
\end{align*}
\end{itemize}
\end{definition}
The operator $\mu$ models a non-invertible function that coerces group elements into exponents. We will refer to this function as a hash function as we symbolically model it as one. We will write $g^x$ for $\dhExp(g,x)$. 
Dougherty's and Guttman's initial theory introduced a sub-sort $\mathit{NZE}$ of $E$ representing non-zero exponents and defined the operator $\dhInv$ only on this sort. We remove this sub-sort in our rewrite theory, and instead ensure that we don't invert the $0$ symbol when solving algebraic equations.  

From the above equations, we can extract rewrite rules by orienting the non-AC equations left-to-right.
Finally, we added further equations derivable from those above to join critical pairs, see  \cite{dougherty}.
We denote this rewrite system by $\rightarrow_{DH}$. Terms that are irreducible with respect to $\rightarrow_{DH}$ are said to be in normal form. 

\begin{definition}
An \emph{irreducible monomial} is a term $m_i$ of the form $\pm (e_1*\ldots*e_l)$, where each $e_i$ is either $x$ or $\dhInv(x)$, for $x$ a variable of sort $E$.
\end{definition}

The same proof as in \cite{dougherty} shows that:

\begin{theorem}
The reduction $\rightarrow_{DH}$ is terminating and confluent modulo AC and normal form terms have the following form.  
\begin{itemize}
\item If $e:E$ is in normal form, then $e=m_1+\ldots +m_k$, where each $m_i$ is an irreducible monomial.
\item If $t:G$ is in normal form, then $t=t_1\cdot \ldots \cdot t_k$, where each $t_i$ is either $v$, $v^{-1}$, $v^m$, or $v^{-m}$, for $v$ a variable of sort $G$ and $m:E$ an irreducible monomial.
\end{itemize}
\end{theorem}



Finally, we introduce the notion of a \emph{root} term. 

\begin{definition}
If $t:G$ is in normal form, with $t=t_1\cdot \ldots \cdot t_n$, then the \emph{root terms} of $t$ are $\roots(t)=\{t_1,\ldots,t_n\}$, and we call $t$ a \emph{root term} if $|\roots(t)|=1$. 
Similarly, if $t:E$ is in normal form, with $t=t_1+\ldots+t_n$, then the \emph{root terms} of $t$ are $\roots(t)=\{t_1,\ldots,t_n\}$, and $t$ is a \emph{root term} if $|\roots(t)|=1$. 
\end{definition}

\subsection{Follow-up rewriting work}
\label{sec:prevwork}
In \cite{dougherty}, Dougherty and Guttman present the term rewrite system just discussed, and prove it provides a decision procedure for equality for a structure representing generic finite fields. Furthermore, they use this rewrite system to derive a symbolic analogue of the discrete log assumption, a computational assumption stating that deducing $e$ from $g^e$ is hard. We rephrase this result in our setting in Section \ref{sec:overview}.  Dougherty and Guttman then sketch how to apply these results in a strand space formalism, which is an alternative operational semantics for the symbolic model. In particular, they show how this discrete log assumption can be used to prove adversarial reachability properties for certain key-exchange protocols. However, their proofs are not mechanized and are based on pen-and-paper arguments. Moreover, they consider only a very restricted class of protocols and only properties stating whether a \emph{passive} adversary can deduce a term.

In follow-up work \cite{Dougherty14Decidability}, Dougherty and Guttman replace their rewrite system with algebraic methods, by interpreting $\mathit{DH}$ messages not as terms in a rewrite theory but as elements in the traditional algebraic structure of fields. 

The authors again work in the strand space formalism. They define a class of protocols, a new adversary who can employ the full algebraic structure of fields, and a language for expressing security properties as protocol goals. They then prove that the resulting combination of protocols, adversary, and goals has the small witness property: if there exists a protocol run that violates a security property, then some run smaller than a computable bound also violates it. The authors provide an algorithm to decide security goals for this class of protocols by testing every protocol run within the bound and verifying whether it satisfies the security goal. This process combines symbolic and algebraic methods: for each run, they build a constraint system, which in turn reduces to a system of linear equations that can be solved by Gaussian elimination.

In practice however, this procedure requires enumerating all possible protocol executions (up to a given bound) and hence is not suitable for automated tools. Instead, we incorporate the idea of combining symbolic and algebraic methods in a backwards search, constraint solving algorithm. 

Furthermore, we remove some restrictions on the class of protocols being modeled. In particular, Dougherty and Guttman require all sent or received Diffie-Hellman group elements to be simple elements in the form $g^x$ for a variable $x$. Although the protocols may compute composite group elements in their local state, e.g. keys, they cannot transmit these elements to the adversary. This is a strong limitation as many protocols that use the full field structure transmit composite elements (e.g., El Gamal). Furthermore, their Diffie-Hellman rewrite theory cannot be combined with arbitrary user-defined equational theories. Finally, they additionally require that whether or not a specific exponent is deemed to be compromised is encoded as an assumption rather than a security goal to be proven.

In \cite{dennisThesis}, Jackson proposes to adapt work from \cite{dougherty} to the framework of an automatic tool, the Tamarin prover. While his approach removes the restrictions from \cite{dougherty} and \cite{Dougherty14Decidability} on the protocols that can be modeled, it only targets passive adversarial reachability properties. Neither correctness proofs nor an implementation were provided. 

In this work, we overcome above shortcomings. Namely, we implement support for full Diffie-Hellman structures in the Tamarin prover, taking inspiration from \cite{dennisThesis}, and using the idea from \cite{Dougherty14Decidability} to combine symbolic and algebraic methods. Our approach hence supports protocols using arbitrary user-defined equational theories with \emph{no} restrictions on the format of messages that are sent or received, as long as all Diffie-Hellman messages belong to the same group. 
We also allow for an active Dolev-Yao adversary and arbitrary security properties expressible in the Tamarin prover, provided the terms appearing in those properties satisfy a certain condition stating that their subterms cannot be inverses of each other (Definition \ref{def:NoCanc}). Finally, we provide an implementation, resulting in the first automated tool that can support protocols using full Diffie-Hellman structures. 



\section{The Tamarin prover}
In Sections \ref{sec:tamarinintro} and \ref{sec:tamarin-int} we provide background on the Tamarin prover: its syntax and its contraint solving algorithm respectively. In Sections \ref{sec:newsorts} and \ref{sec:ext-dep-graph} we describe how we integrate the new equational theory into Tamarin.
\subsection{Models and properties}
\label{sec:tamarinintro}
The Tamarin prover \cite{Tamarin} models protocols as labelled transition systems: 
protocol rules and adversarial capabilities are specified as multiset rewrite rules that define state transitions. The system's state is modeled as a multiset of \emph{facts} (symbols taking terms as arguments) representing the protocol state, adversary knowledge, and network messages.

Terms are elements of an order-sorted term algebra $\mathcal{T}$ that has a unique sort $\mathit{Msg}$ and two sub-sorts $\mathit{Fresh}$ (representing unique messages) and $\mathit{Public}$ (representing publicly known terms), incomparable to each another. This assumes an infinite set of fresh and public names $\mathit{FN}$ and $\mathit{PN}$, and an infinite set of variables $\mathcal{V}$ with an infinite subset of variables for each sub-sort: $\mathcal{V}_f$ (fresh) and $\mathcal{V}_p$ (public). In Tamarin, public variables are annotated with a $\$$\ symbol.  A user-defined signature $\Sigma_{Usr}$ (a set of function symbols) is used to model cryptographic operators, and the function symbols' semantics is defined by an equational theory $\mathcal{E}_{Usr}$. Formally we define the term algebra $\mathcal{T}$ to be the algebra over $\Sigma_{Usr}$, $\mathit{PN}$, $\mathit{FN}$, and $\mathcal{V}$. 

Not all equational theories are supported by Tamarin. Supported theories must satisfy the so-called \emph{finite variant property}, which ensures that there is a general 
unification algorithm modulo $\mathcal{E}_{Usr}$ that returns a \emph{finite} set of most general unifiers \cite{Lundh, ESCOBAR}. 
Note that neither the theory representing Diffie-Hellman operators (Definition \ref{def:dheq}) satisfies this property nor does any theory including distributivity over an $AC$ operator \cite{ESCOBAR,diophantine}.


A multiset rewrite rule is written as follows:
\begin{lstlisting}[escapeinside={*}{*}]
rule rulename:
[Premises(...)]--[Actions(...)]->[Conclusions(...)]
\end{lstlisting}
We refer to the zero or more facts appearing in the premises, labels, or conclusions of a rule as \emph{premise facts}, \emph{action facts}, and \emph{conclusion facts} respectively. 

As is usual in the symbolic setting with a Dolev-Yao adversary, the adversary controls the entire network and sees every message sent.
Tamarin has a special fact \texttt{K(m)}, where \texttt{m} is any term, that encodes the adversary's knowledge (i.e. the adversary knows term \texttt{m}) and the special facts \texttt{In(m)} and \texttt{Out(m)} that encode network messages (i.e. \texttt{m} is received from the network, and \texttt{m} is sent out to the network). 
Adversary capabilities are also modeled by rewrite rules. For example, Tamarin models the adversary reading and sending messages on the network as follows:
\begin{lstlisting}
rule send: [K(m)]--[K(m)]->[In(m)]
rule recv: [Out(m)]-->[K(m)]
\end{lstlisting}
Similarly, Tamarin gives the adversary the same computational capabilities as the other parties running the protocol. In particular, the adversary has access to all the function symbols defined in the protocol. It can apply each function \texttt{f/n} of arity $n$ via rules of the form:
\begin{lstlisting}
rule f: [K(t1),...,K(tn)]-->[K(f(t1,...,tn)]
\end{lstlisting}

All other message deduction rules and network rules can be found in \cite{Tamarin}.
Given a protocol, we denote by $\mathcal{P}$ the set of all protocol rules and adversary capability rules.  The transition system these rules define operates on ground terms: a rule is applied by removing the appropriately instantiated premises from the state and adding the instantiated conclusions. A protocol execution is an alternating sequence of states and instantiated rules, where each state in the sequence is obtained from the previous one by applying the interposed rule. The instantiated action facts in this sequence are the labels of the labelled transitions and they define traces, corresponding to protocol executions. We define security properties over these traces.

More precisely, security properties, called \emph{lemmas}, are expressed as first-order temporal logic formulas over \emph{trace atoms}, where quantification is allowed both on message terms and on timepoints (which in Tamarin are preceded by a \texttt{\#} symbol). Trace atoms are: equality of terms, equality and ordering of timepoints, a label \texttt{ActionName} at a timepoint \texttt{\#i}, written \texttt{ActionName(arguments)@\#i}, or the $\perp$ symbol, together with the usual logical connectives and quantifiers. 
There are some restrictions on the use of quantifiers within formulas. 

Before giving an overview of how Tamarin operates, and proves security properties, we explain how we integrate the full Diffie-Hellman theory into the Tamarin prover. In doing so, we provide a running example of modelling theories in Tamarin. 

\subsection{Introducing new sorts}
\label{sec:newsorts}
We extend Tamarin's term algebra with the two incomparable sorts $G$ and $E$. These two sorts are subsorts of the existing $\mathit{Msg}$ sort, but are incomparable to the existing $\mathit{Public}$ and $\mathit{Fresh}$ sorts. We also introduce a public subsort $\mathit{PubG}$ for $G$ and a fresh subsort $\mathit{FrE}$ for $\mathit{E}$. For technical reasons described in Section \ref{sec:newrules}, we also introduce a subsort $\mathit{varE}$ of the sort $E$. 
We again assume we have infinite sets of names $\mathit{PN}_G$ and $\mathit{FN}_{\mathit{E}}$ and an infinite set of variables $\mathcal{V}_G$ and $\mathcal{V}_E$ with respective infinite subsets $\mathcal{V}_{PubG}$, $\mathcal{V}_{varE}$, and $\mathcal{V}_{Fr_{\mathit{E}}}$. By a $\mathit{DH}$ term, we mean a term of sort $G$ or $E$. We extend the signature $\Sigma_{Usr}$ with the function symbols in $\Sigma_{DH}$
, whose semantics is described by the equational theory $\mathcal{E}_{DH}$. 
We assume that $\mathit{DH}$ terms are always reduced to their normal form after each rule application.  
\begin{align*}
\Sigma & :=\Sigma_{Usr}\cup \Sigma_{DH},\\
\mathcal{E} &:=  \mathcal{E}_{Usr} \cup \mathcal{E}_{DH}.
\end{align*}
We will henceforth work in the term algebra $\mathcal{T}'$, the term algebra over $\Sigma \cup \mathit{PN} \cup \mathit{FN} \cup FN_G\cup \mathit{PN}_G\cup \mathit{FN}_{\mathit{E}} \cup PN_{\mathit{E}}\cup \mathcal{V}\cup \mathcal{V}_G \cup \mathcal{V}_E.$ Similar to terms of sort $\mathit{Fresh}$, distinct terms of sort $\mathit{FrE}$ cannot be equal and, like the terms of sort $\mathit{Public}$, terms of sort $\mathit{PubG}$ are known by the adversary. 

Furthermore, for each operator, we add a rule allowing the adversary to apply it, and we add the following rules allowing the adversary to learn the constant terms and to generate fresh values. 
\begin{lstlisting}
rule FrE: [Fr(f:FrE)]--[K(f:FrE)]->[Out(f:FrE)]
rule const: []--[K(1), K(0)]->[Out(1), Out(0)] 
\end{lstlisting}

As a running example, consider the following simple protocol where Bob sends Alice an encrypted message $m$ using ElGamal.

\begin{center}
\begin{msc}[small values, instance distance=4cm]{ElGamal}
\label{proto1}
\drawframe{no}
\declinst{A}{\footnotesize knows $ska$}{Alice}
\declinst{B}{\footnotesize knows $pka=g^{ska}$}{Bob}
\vspace{-5mm}
\action*{\footnotesize %
fresh $m$, fresh $y$}{B}
\nextlevel[2.5]
\mess{$g^{c_1}=g^y, g^{c_2}=m\cdot pka^y $}{B}{A}
\nextlevel[0.2]
\action*{\footnotesize %
compute $m=g^{c_1*(-ska)}\cdot g^{c_2}$}{A}
\nextlevel[1]
\end{msc}
\end{center}

Note that prior to our work, Tamarin could not model such a protocol, since it did not support group multiplication and the associated equations. Note too that all the previous works mentioned in Section \ref{sec:prevwork} also cannot model such a protocol, nor any other work that we are aware of. 
With our proposed extension for Tamarin, we can model the above protocol as follows, where \texttt{g:PubG} is a constant. 
\begin{lstlisting}[escapeinside={*}{*}]
rule KeyGen:
[Fr(ska:FrE)] -->
  [!PubKey($A, g^ska:FrE), !SKey($A, ska:FrE),  
   Out(g^ska:FrE)]

rule CompromiseKey:
[!SKey($A, ska:FrE)] 
--[Compromised($A)]->
  [Out(ska:FrE)]

rule BobEncrypts:
let pka = g^{ka:E} 
in
[!PubKey($A, pka), Fr(m:FrE), Fr(y:FrE)] 
--[BSent(g^{m:FrE}), SecretB($B, $A, g^m:FrE)]-> 
[Out(<g^y:FrE, g^m:FrE.(pka^y:FrE)>)]

rule AliceReceives:
let m = (g^{c1:E})^(-ska:FrE).g^{c2:E}
in 
[In(<g^{c1:E}, g^{c2:E}>), !SKey($A, ska:FrE) ]
--[AReceived(m), SecretA($A, m) ]->
[]
\end{lstlisting}
\noindent For our running example, we will consider two properties:
\begin{itemize}
\setlength\itemsep{0.8pt}
\item \textit{Executability:} Tamarin finds a trace where both Alice and Bob execute their rule, and Alice receives the intended message.
\item \textit{Secrecy:} If Bob encrypts a message intended to be a secret for Alice, only Alice and Bob should learn it.
\end{itemize}

\noindent We formulate these properties in Tamarin as follows: 

\begin{lstlisting}
lemma executable: 
exists-trace
"Ex msg #i #j. BSent(msg)@i 
               & AReceived(msg)@j
               & not (Ex #l X. Compromised(X)@l)"

lemma secrecy:
"All msg #i B A. SecretB(B, A, msg)@i 
              & not (Ex #l. Compromised(A)@l)==> 
       not (Ex #j. K(msg)@j )"
\end{lstlisting}

Note that this model indeed assumes that all elements are of the form $g^e$ for the same group generator $g$. For example, when Alice receives messages from Bob in \texttt{AliceReceives}, the model explicitly enforces them (by pattern matching) to be of the form $g^{c_1}$ and $g^{c_2}$ for some arbitrary $E$-variables $c_1$ and $c_2$. Observe that an $E$-variable can be instantiated with a term containing the $+$ symbol, so this does \emph{not} exclude $g^{c_1}$ from being a product of different terms. The same holds for Alice's public key $g^{\mathit{ka}}$ in \texttt{BobEncrypts}. It is the user's responsibility to ensure that Bob never uses the $\mathit{ka}$ variable in isolation. A syntactical translation on the compiler's side (translating a variable $\mathit{v_G:G}$ to the term $g^{v_E:E}$) would easily remove this burden. This constraint only misses attacks that use elements from different groups. However this should be an acceptable restriction, since validating group elements in a Diffie–Hellman key exchange is a standard practice. 


\subsection{Dependency graphs modulo full Diffie-Hellman}
\label{sec:ext-dep-graph}

Tamarin has its own specialized data structure for efficiently representing the application of multiset rewriting rules, called 
 \emph{dependency graphs}. Intuitively, these graphs represent sequences of rewrite rule instances that correspond to protocol executions and capture which facts originate from which rules. 
In particular, nodes are labelled with rule instances and two nodes share an edge if a conclusion of the first node matches, modulo the protocol's specified equational theory, a premise of the other. We formally define these graphs in Appendix \ref{sec:appendixdefs}. 

A protocol satisfies an \emph{exists-trace property} $\exists x. \phi$ if there exists an instance of a dependency graph (e.g. a trace or protocol execution) that is a model of the formula $\phi$. It satisfies a \emph{forall-trace property} $\forall x. \phi$ if there does \emph{not} exist a graph that is a model of the formula $\neg\phi$. 

We refer again to \cite{Tamarin} for Tamarin's formal semantics, but intuitively an instance of a graph is a model of a formula if the action facts appearing in the formula appear in the graph and satisfy all timepoint equalities and inequalities expressed by the formula. 


\begin{figure}[h]
\includegraphics[width=8.5cm]{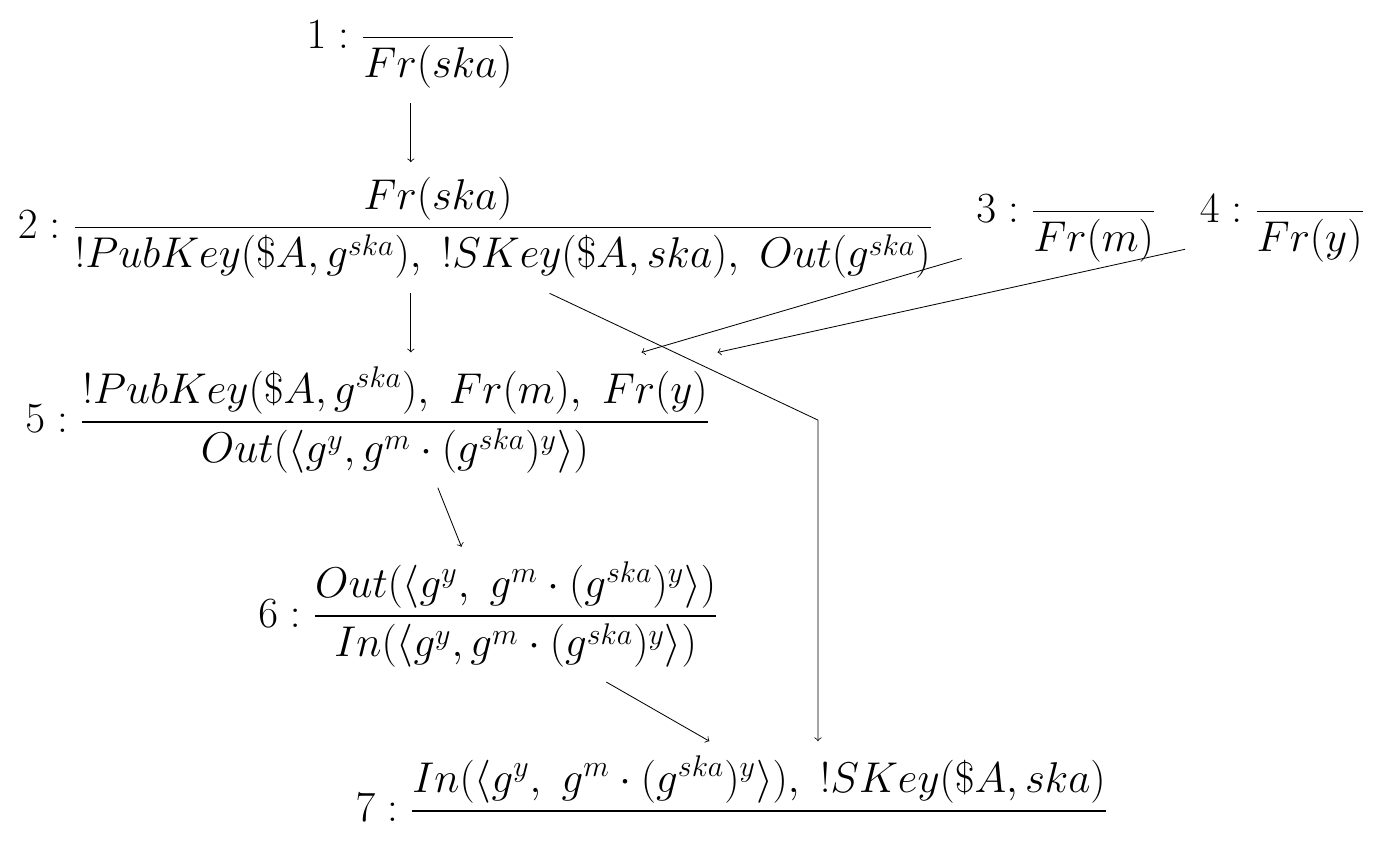}
\caption{A dependency graph of an execution of the ElGamal protocol, where we have omitted the rules' labels.}
\label{fig:depgraph}
\end{figure}

We define $\mathit{DH}$-extended dependency graphs to consider rules instantiated over the term algebra $\mathcal{T}'$. Figure \ref{fig:depgraph} illustrates a DH-extended dependency graph corresponding to an execution of the ElGamal protocol described above.
\begin{definition}
The nodes of \emph{DH-extended dependency graphs} are labelled with rule instances and an edge links two nodes if the conclusion of one matches the premise of the other, now modulo $\mathcal{E}_{Usr}\cup\mathcal{E}_{DH}$.  
\end{definition}

This is well defined because we can check whether two \emph{ground} terms in $\Sigma$ are equal, since $\Sigma_{Usr}$ and $\Sigma_{DH}$ have disjoint signatures, and we can check the equality modulo AC of ground terms within each equational theory by considering their normal forms.

\subsection{Constraints and constraint solving}
\label{sec:tamarin-int}

As mentioned, by negating the forall-trace formulas, it suffices to determine the existence of protocol traces that are models for a given formula. This problem is undecidable, so Tamarin's search might not terminate. 

Tamarin searches backwards, extracting a set of \emph{constraints} on graphs from the targeted formula. Similar to atoms, constraints are action facts that must appear in traces (e.g. $Fact@i$), term equalities and timepoint orderings, and also rules that must be applied at certain timepoints ($i:rule$). To construct a solution to a constraint system (a finite set of constraint sets), 
Tamarin uses a constraint solving relation $\rightsquigarrow$, where $C_1 \rightsquigarrow C_2$ reduces the constraint system $C_1$ to the constraint system $C_2$ by inferring new constraints and performing case splits. 
We write $c_1,\ldots,c_n,\Gamma$ to denote $\{c_1,\ldots,c_n\}\cup \Gamma$. 
The rule $$\begin{crule}
{c, \Gamma } { \{ \Delta_i, f_i(\Gamma) \}_{1\leq i\leq n} } {RuleName} 
\end{crule}$$
denotes that
$$\{c\}\cup \Gamma \rightsquigarrow \{ \Delta_i \cup f_i(\Gamma),\ldots, \Delta_n \cup f_n(\Gamma)  \}.$$

For example, denoting with \textit{acts(ri)} all the action facts (i.e. neither including the premises nor the conclusions) of rule \textit{ri}, Tamarin's constraint solving rule 
$$
\begin{crule}
{\mathit{Fact}@i, \Gamma}{\{\mathit{Fact}=g, i:\mathit{ri}, \Gamma\}_{\mathit{ri}\in \mathcal{P}, g \in \mathit{acts}(ri)}}{S_@} 
\end{crule}
$$
solves a constraint $\mathit{Fact}@i$ by performing a case split over all rules $ri$ and their actions $g$ that might be equal to $Fact$. In each case, a constraint $\mathit{i:ri}$ (indicating that an instance $\mathit{ri}$ of a protocol rule must be applied at timepoint $i$) is added. Additionally,  for each of the actions $g$ of $ri$, an equality constraint between $g$ and $Fact$ is added.

Another example is the constraint solving rule that introduces an equality constraint, this time between the premises and conclusions of rules. More precisely, denoting by $prems(ri)$ and $concs(ri)$ the premise facts and conclusion facts of rule $ri$, the rule
$$
\begin{crule}
{i:ri, \Gamma}{\{i:ri,\ j:ru,\ pi=cj, \Gamma\}_{\substack{ru\in \mathcal{P}, pi \in prems(ri),\\ cj\in concs(ru)}}}{S_{Prem}} 
\end{crule}
$$
solves a constraint $i:ri$ by performing a case split over each rule $ru$ checking whether its conclusions are equal to some premise of $ri$. In each case, an equality constraint between rule conclusions and premises is added.

Equality constraints are then solved by another rule $S_{=}$ that performs unification.
$$
\begin{crule}
{f=f', \Gamma}{\{\sigma(\Gamma)\}_{\sigma\in \unify_{\mathcal{E}_{Usr}}(f,f')}}{S_{=}}.
\end{crule}
$$
Note that 
this unification rule can also be used to solve timepoint equalities.  

The full list of Tamarin's constraint solving rules can be found in \cite{Tamarin} along with a proof of the following theorem.

\begin{theorem}
\label{thm:sound}
The relation $\rightsquigarrow$ transforming $\Gamma \rightsquigarrow \{\Gamma_1,\ldots,\Gamma_k\}$ is sound
and complete in the following sense: 
\begin{itemize}
\item \emph{Soundness}: The union of the models (i.e. dependency graphs) of the $\Gamma_i$ is a subset of the models of $\Gamma$. 
\item \emph{Completeness}: the set of models of $\Gamma$ is a subset of the union of the models of the $\Gamma_i$. 
\end{itemize}
\end{theorem}

Finally, to find a model of the constraint system, Tamarin searches for a constraint reduction $\{\phi\} \rightsquigarrow ^* \emptyset$ or a constraint reduction $\{\phi\}\rightsquigarrow ^* \Delta$ such that there is a constraint system $\Gamma\in\Delta$ for which we can directly confirm that it has a model. In the first case, Tamarin proves that $\phi$ is not satisfiable (if $\phi$ represents the negation of a security property, this proves the property). In the second case, Tamarin extracts a model (hence a dependency graph) $dg$ of $\phi$ from $\Gamma$ (if $\phi$ represents the negation of a security property, this represents a counterexample to the property). In both cases, such a constraint system is defined to be solved.

In Section \ref{sec:newrules}, we show how we extend Tamarin's current constraint solving algorithm by defining new constraint solving rules for terms of sort $\mathit{DH}$. 
We leave Tamarin's constraint solving rules unchanged on terms of other sorts. In Section \ref{sec:combination}, we show how we can merge approaches for terms that contain both sorts. Observe that in our protocol model of ElGamal, we have already used the pair function over terms of sort $\mathit{DH}$, so this integration is crucial for usability.

\textbf{Examples.} Returning to our example, when trying to prove the lemma \texttt{executable}, Tamarin (as it currently stands) produces the constraint set\footnote{In this and the following example, we present simplified constraint sets. For clarity of exposition, we have omitted checking that there is no $Compromised$ action fact in the trace. }:
\begin{exbox1}
\label{ex:exec}
$$\{BSent(msg)@i,\ AReceived(msg)@j\}.$$
\end{exbox1}

Let us see why Tamarin currently cannot solve such a constraint. Since $BSent$ only appears in the \texttt{BobEncrypts} rule, Tamarin (using $S_{@}$ and $S_{=}$) will replace the variable $\mathit{msg}$ by the term $g^{\sim m}$ in the entire constraint system, and will then try to solve the $AReceived(g^{\sim m})$ constraint. This fact only appears in the \texttt{AliceReceives} rule, and, again via the $S_{@}$ and $S_{=}$ rules, Tamarin would try to unify $g^{\sim m}$ and $(g^{c_1})^{-ska}\cdot g^{c_2}$. However, as mentioned, we do not have a unification algorithm for this!  

When trying to prove the lemma \texttt{secrecy}, Tamarin tries to find a trace for its negation. After replacing the variable $\mathit{msg}$ by $g^{\sim m}$ (since $SecretB(B,A,g^{\sim m})$ only appears in one rule), it will try to solve the remaining constraint system:
\begin{exbox2}
\label{example:secrecy}
$$\{K(g^{\sim m})@j\}.$$
\end{exbox2}

Solving adversarial facts is more complicated. Essentially, Tamarin would try to see if $g^{\sim m}$ can be unified with a term coming from an $Out$ fact or with the output of a function symbol applied on other terms in the adversary's knowledge, again learnt from $Out$ facts or via other function symbol (see rules \texttt{send} or \texttt{f}). 
Again, this approach is not possible with the full DH theory, as we cannot check if two terms are unifiable. 

\section{Indicators and non-cancellation}
\label{sec:overview}

  
As illustrated in the previous examples, the two views we need to cover are: (i) honest protocol execution, and (ii) adversarial behaviour. 


For honest agents, behavior is strictly described by the protocol rules. To unify the premise of one rule with the conclusions of others, we exploit the non-cancellation assumption: every root term in the premise must also appear in the conclusion. Rather than unifying full terms, we first unify premise root terms with conclusion root terms in a simplified theory with a dedicated unification algorithm. We then use algebraic reasoning to check whether the instantiated conclusion terms can be combined, possibly canceling surplus terms, to reconstruct the premise.

Adversarial behaviour is more complex as the adversary may freely apply all $DH$ operators. 
To decide whether a target term can be deduced, we classify the exponents in the target term as known or unknown, isolating the secret components that must be obtained from the network. We unify these secret components with network messages in the simplified theory, and again use algebraic reasoning with the known exponents to reconstruct the target term.


A key step of the approach is determining which exponents are known to the adversary. Given a constraint system, we can keep track of the exponents leaked, i.e. made available to the adversary, within it:

\begin{definition} 
\label{def:basis-ind}
Let $\Gamma$ be a constraint set. Given a timepoint $i$, we call the set of exponents that are known by the adversary at timepoint $i$ the \emph{adversarial leaked set}: $$e\in L^{\Gamma}_{i} \iff {\Large \substack{ K(e)\in acts(ru) \text{ for some contraint }\\ j:\mathit{ru} \in \Gamma \text{ with }j<i}}$$ 
\end{definition}
These sets correspond to the non-basis and basis sets in Dougherty and Guttman's terminology. 
Note that if $e \in L^{\Gamma}_i$, then $e\in L^{\Gamma}_j$ for all $j>i$. 
The set $L^{\Gamma}_i$ is easily computable from the current constraint system $\Gamma$. In practice we store exponents \emph{and} the timepoint at which they are learned in our constraint system. Each time we add a rule with action $K(e)$ at timepoint $j$, we add $(e,j)$ to the basis set of our constraint system and $L^{\Gamma}_i$ consists of all exponents $(e,j)$ in this set such that $j\leq i$. 

As we explain in Section \ref{sec:adv}, to classify the remaining exponents appearing in our target term, we will follow Tamarin's general strategy of using backwards search. 
For each exponent variable $e$ appearing in our term, if it is not already in the leaked set $L^{\Gamma}_i$ of the current constraint system $\Gamma$, we introduce the constraints $K(e)@j_e$ and $j_e < i$, and attempt to solve these constraints. 
This works well when the exponent terms appearing in the term are known to be fresh, since in that case the unification sub-queries are extremely simple: fresh terms can only be unified with other fresh terms. In the more general case though, Tamarin may fail to terminate when trying to classify the exponent as secret or leaked. In general, we delay solving the $K(e)@j_e$ constraints when $e$ is not of sort fresh. 
In most protocols that use Diffie-Hellman groups, the exponents used to exchange keys are randomly chosen, hence we believe this limitation does not widely impact the usability of our approach. 


Given a leaked set $L$ and a term, we now define its `secret' parts, one for each of its root terms. 

\begin{definition}
\label{def:ind}  For a root term $t$, a set of $E$-variables $L$, and corresponding complement $S=\mathcal{V}_E\setminus{L}$, 
we call its \emph{indicator}, denoted by $\Ind_L(t)$, the sub-term of $t$ needed to construct $t$ using only the exponents in $L$. 
\end{definition}
Observe that the indicator is defined only on \emph{root} terms and hence it is unique (for the given set $L$). 
For example, if $L=\{y,w\}$, and $t=g^{x*y}$, then $\Ind_L(t)=g^x$. When a term is already deducible from $L$, we define its indicator to be $1$, e.g., for the same $L$, $\Ind_L(g^w) = 1$. 

In general, $Ind_L(t)$ only contains variables that belong to $S$. An exception is when $t$ contains a hashed term that contains elements both from $L$ and $S$. Since $\mu$ models a one-way function, as soon as it contains one secret exponent, we must consider the entire term secret, e.g. $Ind_L(g^{\mu(x*w)}) = g^{\mu(x*w)} $. We give the precise recursive definition of the indicator function in Appendix \ref{app:DHdefs}.

The following theorem states that new indicators cannot be created by function applications. It is a reformulation of Dougherty and Guttman's Indicator Theorem from \cite{dougherty}, rephrased with Tamarin-specific notation. 
\begin{definition}
Let $T$ be a set of terms. We define $\Gen_{DH}(T)$ to be the least set of terms including $T$ that is closed under the application of function symbols in $\Sigma_{\mathit{DH}}$.
\end{definition}
 


\begin{theorem}
\label{thm:cdh} Let $T$ be a set of terms. 
Let $L$ be a set of $E$-variables and $S=\{e\ |\ e:E\} \setminus L$. Assume that for each $e\in T$ of sort $E$: $Ind_L(e)=1$. 
Then the following hold.

\begin{enumerate}
\item Every $e\in \Gen_{DH}(T)$ of sort $E$ also satisfies $Ind_L(e)=1$.
\item If $u\in \Gen_{DH}(T)$ is of sort $G$ and $z\in \{ Ind_L(x)\ | \ x\in \roots(u) \}$, then for some $t\in T$ we have that $z\in \{ Ind_L(x)\ | \ x \in \roots(t) \} $. 
\end{enumerate}
\end{theorem}
This theorem is an analogue of the discrete logarithm assumption. 
From Theorem \ref{thm:cdh}, once we have established the leaked set, given a term $t$, it suffices to find \emph{one} indicator of $t$ that cannot be reconstructed by the adversary, or an agent, to conclude that the whole term cannot be known. Indeed, by Theorem \ref{thm:cdh}, all indicators must be known to construct the term. 
If, however, we cannot exclude any indicator, we must determine whether these indicators can be combined correctly to obtain the target term. We will elaborate on how this can be done with algebraic methods in Section \ref{sec:C=}. 

Observe that this strategy does not consider variable substitutions. Since $\Sigma_{DH}$ contains reducible functions, substitutions could cancel out particular root terms, invalidating the above reasoning. To apply Theorem \ref{thm:cdh} in Tamarin's unification-based setting, we need the extra assumption that particular root terms do not cancel each other out. We define the condition $\NoCanc(X,Y)$, which states that the root term $X$ cannot be cancelled out by the root term $Y$. 
\begin{definition}
\label{def:NoCanc}
\begin{align*}
&\NoCanc(X,Y) \Leftrightarrow \\
&\forall \theta \in \mathcal{S}.\ \theta(X)\neq 1 \land \dhInv(\theta(X)) \neq \theta(Y),
\end{align*}
where $\mathcal{S}$ is the set of all substitutions for the $\mathit{DH}$ system. 
\end{definition}

For example, $\NoCanc(X,Y)$ is satisfied when both $X$ and $Y$ only contain variables of sort $\mathit{FrE}$. 
Another condition for which $\NoCanc(X,Y)$ holds concerns the $\mu$ operator.

\begin{theorem}
\label{thm:NoCanc}
The condition $\NoCanc(X,Y)$ holds for terms $X$ and $Y$ if:
\begin{enumerate}
\item $X =g^{b*\mu(g^t)}$ and $Y=g^{a*t} $, or
\item $X = g^{a*\mu(g^{t_1})}$ and $Y = g^{b*\mu(g^{t_2})}$, 
\end{enumerate}
where $a$ and $b$ are terms that only contain distinct variables of sort $\mathit{FrE}$ and no $\dhInv$ function application.  
\end{theorem}
We provide the proof in Appendix \ref{app:DHdefs}. Essentially, this follows from the fact that fresh variables cannot be instantiated with function applications. Hence in $1)$ only the $t$ term can be instantiated with an $\dhInv$ term, creating an infinite recursion since $t$ appears in both $X$ and $Y$. In $2)$ the $\dhInv$ symbol can only appear in a $\mu$ term, which however has no associated equations, and hence cannot be `pulled out'.

Note that the definition excludes $X$ from cancelling itself out (by excluding $\theta(X)=1$) and hence is asymmetric. To consider both directions, we define $\NoCanc$ on \emph{all} the root terms of a term simultaneously. 
\begin{definition}
Abusing notation, for a term $t$, the condition $\NoCanc(t)$ holds if $\NoCanc(X,Y)$ holds for all $X\in \roots(t)$ and $Y\in \roots(t)$, with $X\neq Y$.
\end{definition}

As mentioned, often a much weaker requirement suffices. If we know that $X_1$ cannot be cancelled out by the other root terms of a term $t=X_1\cdot ... \cdot X_n$, then it suffices to prove that $X_1$ cannot be deduced to conclude that the entire term $t$ cannot be deduced. In this case, we do not need the non-cancellation property on the \emph{other} root terms. Hence we also define the weaker condition for a root term $X\in \rt(t)$ of a term $t$: 
\begin{definition}
\label{def:NoCanc2}
\begin{align*}
\NoCanc'(X,t) \Leftrightarrow \forall Y\in rt(t)\setminus \{X\}.\ \NoCanc(X,Y).
\end{align*}
\end{definition}

Before applying a constraint solving rule concerning a given term $t$, our algorithm checks if these sufficient conditions on its root terms are satisfied. If they do not hold, we store the (minimal) set of pairs of root terms that need to satisfy the $\NoCanc$ condition to guarantee the correctness of our approach and output these assumptions. 
In general, our approach is guaranteed to work for every protocol and security property such that every $\mathit{DH}$ term $t$ appearing as an argument of an action fact in the property statement either already satisfies $\NoCanc(t)$ or the arguments of the corresponding action facts in the protocol satisfy the $\NoCanc$ property. %

 \section{Alternative strategy for DH terms}
 \label{sec:newrules}
We extend \emph{all} of Tamarin's current constraint solving rules to be applicable to all constraints, except for the rule $S_{=}$, and the rules that introduce equality constraints. Recall that we cannot apply the rule $S_{=}$ to constraints containing $\mathit{DH}$ terms as it requires unification in the $\mathit{DH}$ rewrite system. The rules introducing equality constraints are: (1) $S_@$, which tries to match an action fact coming from a security property (lemma) with an action fact appearing in a rule label and (2) $S_{Prem}$, which tries to match a rule's premise fact with another rule's conclusion fact.  

We therefore introduce new constraint solving rules that complement these three rules for constraints containing Diffie-Hellman terms. Our new rules incorporate \emph{algebraic} methods and replace unification by interpreting our rewrite theory as a traditional algebraic structure (i.e. a real field extension) and leveraging mathematical methods such as Gaussian elimination or Buchberger's algorithm.
 


To produce mathematical equations, 
define a function that generalizes substitutions by variables of sort $E$. 

\begin{definition}
\label{def:gen}
Let $\sigma$ be a substitution whose domain contains variables of sort $E$ or its subsorts. We define the substitution $gen(\sigma)$:
$$gen(\sigma) (v) = \begin{cases}
\sigma(v)  &  \text{ if }v:\mathit{FrE} \\
\sigma(v)+ Y_v  &  \text{ if }v:E \land \sigma(v)\neq v \\
Y_v  &  \text{ if }v:E \land \sigma(v)= v
\end{cases}$$
where $Y_v$ is a new variable of sort $\mathit{E}$.
\end{definition}
Finally, we define a simplified equational theory for which a unification algorithm is known. 
\begin{definition}
We set $\varepsilon_{simp}$ to be the equational theory obtained from $\Sigma_{DH}$ by deleting any equation containing the operators $+$, $-$, and $0$.
\end{definition}

We now show our approach for action facts $Fact(\cdot)$ that are not the adversarial $K$ fact, covering honest protocol execution.
 We will introduce constraint solving rules for adversarial $K$ facts in Section \ref{sec:adv}. 

Assume that $\NoCanc(t)$ holds for a certain term $t$. As a starting point, we also search for conclusions/actions of the form $Fact(h)$ in the protocol rules. Let $L$ be the leaked set of the rule to which $Fact(h)$ belongs. To see if $t$ and $h$ are unifiable, we check if each root term of $t$ can be unified in $\varepsilon_{simp}$ with a root term of $h$. Note that Tamarin currently supports unification in  $\varepsilon_{simp}$, returning a complete set of most general unifiers per unification query. Since we generalize the returned substitution, in most cases (see Appendix \ref{app:newrules}), we do not need the complete set and instead consider one unifier per query. Note that for unification queries in $\varepsilon_{simp}$, a complete set of unifiers can consist of more than 40 unifiers. 

In particular, let $t_1,\ldots,t_n$ be the root terms of $t$, and $h_1,\ldots,h_m$ the root terms of $h$. Consider all possible $n$ (possibly repeated, by introducing a $+$ symbol) combinations of the subterms $h_i$. 
We solve the equations 
$$\{t_i =_{\varepsilon_{simp}} h_{j_i}\}_{i=1}^{n}$$
 for $i=1,\ldots,n$, $j_i \in \{1,\ldots,m\}$. 
If there is an $i$, for which the above equation has no unifier for some $j_i$, we conclude that the terms are not unifiable. 
If however there is a unifier $\sigma$ that unifies each root terms of $t$ with some root term of $h$ 
we then solve the algebraic equation $$gen(\sigma)(t) \dheq gen(\sigma)(h)$$ for the unknowns $Y_i$ introduced by the function $gen$. 
Note that we will henceforth use the symbol $\dheq$ to denote equality constraints between terms containing unknown variables. 

In Appendix \ref{app:newrules}, we define the constraint solving rules that perform these steps. We name our new rules with the letter $C$ to distinguish them from existing constraint solving rules (named with the letter $S$):

\begin{enumerate}[itemsep=0pt, topsep=1pt, partopsep=0pt]
\item[(i)] $C_{@}$ replaces $S_{@}$ for $\mathit{DH}$ terms and performs a case split, generating a $gen(\sigma)(t) \dheq gen(\sigma)(h)$ constraint for each $\sigma$ obtained in the procedure above;
\item[(ii)] $C_{Prem}$ replaces $S_{Prem}$ for $\mathit{DH}$ terms and is defined in the same way as $C_{@}$, but with premises and conclusions instead of action facts;
\item[(iii)] $C_{=}$ replacing $S_{=}$ for $\mathit{DH}$ terms, which solves the $gen(\sigma)(t) \dheq gen(\sigma)(h)$ constraints (see Section \ref{sec:C=}). 
\item[(iv)] $C_{var}$ replaces $S_{@}$ and $S_{Prem}$ if $t$ (or $h$) is a variable. It directly instantiates the variable $t$ with the term $h$ (or respectively $h$ with $t$), skipping the above rules. 
\end{enumerate}
Observe that these rules replace their $S$ counterparts \emph{only} for $\mathit{DH}$-sorted terms. For constraints that contain terms of sort $\mathit{Msg}$, we apply the usual $S$ rules.

In Appendix \ref{app:newrules}, we also show that this extended constraint solving relation is sound (according to the previous definition) with respect to $\mathit{DH}$-extended dependency graphs, and \emph{under some assumptions} (i.e. non-cancellation of subterms of $t$) it is also complete. This guarantees the correctness of Tamarin's verdicts and hence whether a property $\exists x. \phi$ or $\forall x. \phi$ is verified will be determined by whether the constraint systems $\{\phi \}$ or $\{\neg \phi \}$ can be solved. 


\textbf{Examples. } 
Returning to Example 1 from in the previous section, we were trying to solve the constraint set:
\begin{exbox1}
\label{ex:exec}
$$\{\mathit{BSent}(\mathit{msg})@i,\ \mathit{AReceived}(\mathit{msg})@j\}.$$
\end{exbox1}
The only rule that contains a \textit{BSent} action fact is \texttt{BobEncrypts}. Since $\mathit{msg}$ is a variable, we can apply $C_{var}$ to replace $\mathit{msg}$ with $g^{\sim m}$ and obtain:
\begin{exbox1}
\label{ex:exec}
\vspace{-2mm}
\begin{align*}
& \{BSent(g^{\sim m})@i, i: BobEncrypts\\
 & AReceived(g^{\sim m})@j\}.
 \end{align*}
\end{exbox1}
The term $t = g^{\sim m}$ is not a variable, so we must apply $C_{@}$ to solve the constraint $\mathit{AReceived}(g^{\sim m})@j$. Observe that $\NoCanc(g^{\sim m})$ trivially holds. The only rule with the action fact \textit{AReceived} is \texttt{AliceReceives}, with argument $h = (g^{c_1})^{-ska}\cdot g^{c_2}$. 
We try to unify the root term $g^{\sim m}$, in the simplified equational theory $\varepsilon_{simp}$, with a root term of $h=(g^{c_1})^{-ska}\cdot g^{c_2}$:
\begin{align*}
&g^{\sim m} =(g^{c_1})^{-ska} \text { or }\\
&g^{\sim m} =g^{c_2}. 
\end{align*} 

Choosing the second equation and unifying in $\mathcal{E}_{Simp}$ gives $\sigma= \{ c_2\mapsto \sim m\}$. The substitution $\sigma$ leaves the other variables of our system unchanged. We generalize (see Definition \ref{def:gen}) to $gen(\sigma)= \{ c_2 \mapsto  (\sim m) +Y_1, c_1\mapsto Y_2 \}$, producing the constraint system:

\begin{exbox1}
\label{ex:exec}
\vspace{-1mm}
\begin{align*}
& \{BSent(g^{\sim m})@i, i: BobEncrypts \\
 & AReceived(g^{\sim m})@j,  j: AliceReceives \\
 & g^{Y_2(-ska)}\cdot g^{(\sim m)+Y_1} \dheq \sim g^{\sim m} \}.
 \end{align*}
\end{exbox1}


\subsection{Algebraic Methods}
\label{sec:C=}

In this section, we describe the rule $C_=$, which excludes or finds a substitution for all the $Y$ variables 
that makes the terms $t$ and $h$ equal. 
 
 For honest executions, all variables appearing in $\sigma(t)$ and $\sigma(h)$ can be manipulated. 
 However, for adversarial behaviour, some exponents are generally unknown to the adversary (and hence cannot be inverted, etc.). To distinguish between the two, for the rest of the section, we will work with a fixed leaked set $L$ and its complementary set, the secret set $S$.  We denote them by $L= \{l_1,\ldots,l_r\}$ and $S=\{s_1,\ldots,s_l\}$ respectively. For honest executions, however, we set $S=\emptyset$.
 
Recall that we assume all terms to be of the form $g^e$, for a fixed group generator $g$. 
 Hence we can focus just on the expressions of the exponents of terms $t=g^{e_t}$ and $h=g^{e_h}$. By abuse of notation, we let $t=e_t$ and $h=e_h$ denote the terms in the exponents. 

Since exponents in $S=\{s_1,\ldots,s_l\}$ are secret, we treat them as abstract variables: for example, we cannot compute their inverses, nor multiply them with other terms. We cannot instantiate our unknowns  $Y_i$ with secret values. In contrast, we can fully manipulate elements in $L= \{l_1,\ldots,l_r\}$. 
We can hence interpret $t$ and $h$ as a polynomials with variables in $S=\{s_1,\ldots,s_l\}$ and coefficients in the polynomial field $\mathbb{Q}(l_1,\ldots,l_r)[Y_i]$. This is similar to what is done in \cite{Dougherty14Decidability}. 



The function that interprets symbolic DH terms as elements in the field of rational functions in $r$ variables $\mathbb{Q}(l_1,\ldots,l_r)$ is a straightforward translation, which we define recursively in Appendix \ref{app:algdef}. Clearly, we can define this translation only from the set $\Gen_{DH}({l_1,\ldots,l_r})$, since different variable names do not have counterparts in $\mathbb{Q}(l_1,\ldots,l_r)$. 

\begin{definition}
Given a set of terms $\{s_1,\ldots,s_l\}$, the function $\toPoly(t)$ interprets a term $t_i\in \Gen_{DH}({l_1,\ldots,l_n,s_1,\ldots,s_l,Y_i})$ as a polynomial with coefficients in $\mathbb{Q}(l_1,\ldots,l_n)[Y_i]$ and variables $s_1,\ldots,s_l$, by naturally interpreting each $s_i$ appearing in $t$ as a variable, and the rest as coefficients in $\mathbb{Q}(l_1,\ldots,l_n)[Y_i]$.
\end{definition}

For example, for $L=\{x,y\}$ and $S =\{u,\dhInv(v),w\}$ we have $$\toPoly(x**\dhInv(y)**w**Y_1 + u**x - u**y**Y_2 + \dhInv(v)**w)$$ $$= \cfrac{x}{y}\cdot Y_1 \cdot S_3 + (x- y\cdot Y_2)\cdot S_1 + S_2\cdot S_3,$$
which is an element of $\mathbb{Q}(x,y)[Y_1,Y_2][S_1,S_2,S_3]$, where we have changed the variable names $u,\dhInv(v),w$ to $S_1,S_2,S_3$ for clarity. Note that since they are just variables, their names are irrelevant.

Returning to our generated constraints $gen(\sigma)(t)\dheq gen(\sigma)(h)$, we would like to answer the following question:
Does there exist a substitution $\sigma$ instantiating the $\{Y_i\}_{i=1,\ldots,n}$ variables such that $$\toPoly(gen(\sigma)(h)) = \toPoly(gen(\sigma)(t))?$$ 

Defining $H:= \toPoly(gen(\sigma)(h_i))$, and $T := \toPoly(t)$, this corresponds to solving the equation with coefficients in $\mathbb{Q}(l_1,...l_r)$ for the given variables $\{Y_i\}_i $:

\begin{equation}
\label{eq:DHActionFacts}
H(s_1,\ldots,s_r) = T(s_1,\ldots,s_r).
\end{equation}

Matching each coefficient of the polynomials, this gives a system of equations in the variables $\{Y_i\}_i$. If we assume $H$ and $T$ are linear in the $Y_i$, this system of equations can be solved via Gaussian elimination. For the case of non-linear equations, one could use more powerful mathematical tools. For example, Barthe et al., in \cite{barthe18}, show that this can be solved using Buchberger’s algorithm for Gröbner basis computations.
We have currently only implemented the Gauss elimination strategy in the Tamarin prover as our case studies all produced linear equations. 

Assume there is a solution to Equation (\ref{eq:DHActionFacts}), and let $\{(y^{(1)}_1,...y^{(1)}_n)$,\ldots,($y^{(m)}_1,\ldots,y^{(m)}_n)\}$ be a basis 
of the solution space to the Equation (\ref{eq:DHActionFacts}). 
Let $\sigma^j = \{Y_i\mapsto y^{(j)}_i\}$ be the substitution obtained by interpreting each $y^{(j)}_i$ as an abstract term. 
The rule that solves the DH-equality constraints is then as follows:
\begin{rul}[$C_{=}$]
\namedthmlabel{rul:C=}
If Equation (\ref{eq:DHActionFacts}) obtained from the constraint $t \dheq h$ admits a solution, let $\{\sigma^j\}$ be the set of substitutions corresponding to a basis of the solution space. 
$$\crule{t\dheq h,\ \Gamma}{ \{ \sigma^j(\Gamma) \}_j  }{C_{=}}.$$
\end{rul} 

We show in Appendix \ref{sec:proofsrules} that Tamarin's contraint solving relation remains sound when adding this rule and it is also complete under the non-cancellation assumption on $t$. 

The constraint system from Example 1 $$g^{(c_1 + Y_2)* (-ska)}\cdot g^{(\sim m)+Y_1}\dheq g^{(\sim m)}$$ produces the equation:
$$Y_2* (-ska) + Y_1 + (\sim m)= (\sim m).$$

Since $S=\emptyset$, this is the only equation we need to solve, with solution: 
$Y_1 = ska*Y_2.$ Introducing a new fresh variable $v_{Y_2}:E$, the basis of the solution space is $\langle ska*v_{Y_2}, v_{Y_2}\rangle $ corresponds to the substitution

$\sigma = \{Y_2\mapsto v_{Y_2}, Y_1\mapsto ska*v_{Y_2}\},$
and thus the constraint system
\begin{exbox1}
\label{ex:exec}
\begin{align*}
& \{BobSent(g^{\sim m})@i, i: BobEncrypts \\
 & AliceReceived(g^{\sim m})@j,  \sigma(j: AliceReceives) \}.
 \end{align*}
\end{exbox1}

The next steps in solving this system would be to solve the premises of the rule $\sigma(j: AliceReceives)$, given by $In(\langle g^{v_{Y_2}},g^{ska*v_{Y_2}}\cdot g^{\sim m}\rangle)$. We illustrate in the following section the optimizations we make for $In/Out$ terms and our approach for combined terms. 

\begin{remark} 
In the above we have omitted two major points. First, by interpreting symbolic variables as elements in $\mathbb{Q}$, we are assuming that we can invert them. However, we should not be inverting terms that are $0$. So far our implementation only covers the case where the returned substitution only inverts expressions that cannot reduce to $0$ (such as $\mu$ terms or terms containing $\mathit{FrE}$ variables) - as was the case for our examples. As future work, when an element needs to be inverted, we would need to consider two cases: (i) interpreting the element as a non-zero, we can proceed as normal, (ii) considering the case where the element is $0$, and recompute the equation to solve with this term instantiated to $0$ and start again. 

Second, our translation from symbolic terms to $\mathbb{Q}(L)$ needs to account for terms that contain the $\mu$ operator. Therefore, following the approach in \cite{Dougherty14Decidability}, we introduce abstract variables replacing $\mu$ terms and add them to $L$. We then split over all possible cases of whether any of these $\mu$-terms are equal. In particular, for each case where we assume $\mu(t_1)=\mu(t_2)$ for some terms, we first solve the equation $t_1=t_2$.   
\end{remark}



\subsection{Adversarial premise facts}
\label{sec:adv}
As mentioned in Section \ref{sec:tamarinintro}, Tamarin has a default Dolev-Yao adversary that can read everything on the network and send messages on the network. It can use the given equational theories, and hence its knowledge set is always determined by the set of all terms sent in the network, all the terms of public sorts and all terms deducible from these terms by applying the equational theories' function symbols. Tamarin has a smart way to avoid infinite loops, e.g. re-encrypting and decrypting continuously, which we do not describe here, but see \cite{Tamarin}. 

For any rule \textit{i:ru} containing a premise fact $K(\cdot)$, we perform an optimization in our implementation, bypassing some applications of the Rule \ref{rul:CPrem}. Recall that for each operator $f$ of arity $n$ in our $\mathit{DH}$ theory, we introduced a message deduction rule \texttt{f:[K(x\_1),\ldots,K(x\_n)] --> K(f(x\_1,\ldots,x\_n))}. Each such rule produces a conclusion $K(\cdot)$ and hence applying the Rule \ref{rul:CPrem} repeatedly without a suitable heuristic leads to an infinite search space. To address this issue, our implementation does \emph{not} include the message deduction rules corresponding to the DH operators among our protocol rules. Instead we represent the adversary's ability to apply these operators directly when passing to the algebraic setting. 

In particular, we introduce the new constraint solving rules \ref{rul:CPremK} 
and \ref{rul:C=K} that replace the Rules \ref{rul:CPrem} and \ref{rul:C=} in the case where the premise we want to solve is $K(t)@i$ for a term $t$ for which $\NoCanc(t)$ holds. 

\subsubsection{General terms}
We first assume that the term $t$ is \emph{not} a variable. 

As a first step, we must determine whether the adversary can deduce all the indicators respect to the adversarial leaked set of each root term of $t$. Recall that the indicator of a root term represents the minimal information required by the adversary to construct that term from its current knowledge. If the adversary cannot deduce all such indicators, we can immediately conclude that it cannot deduce the term $t$ itself.

Since the adversary can manipulate and combine different terms he knows, we do not unify the indicators of the root terms of $t$ with the root terms of a single term in an $Out$ facts. Rather, we unify the indicators of $t$ with the indicator of root terms of possibly different $Out$ fact, as long as they are of the correct sort. In particular, if $t$ has $n$ root terms, we consider all possible $Out(h_1),\ldots,Out(h_n)$ that appear as the conclusions of protocol rules and try to unify each indicator of $t$ with the indicator of a root term of each $h_i$:
 $$\{Ind_L(t_i)=Ind_L(h_i)\}_{t_i\in rt(t), Ind_L(t_i)\neq e_G}.$$

For this, it is necessary to establish the adversary's knowledge set $L$ at timepoint $i$. In particular, for each choice of $h_1,\ldots,h_n$, we only need determine whether the adversary could have learnt the exponents appearing in $t$, $h_1$,\ldots, $h_n$ previously.  To do so, we consider the current leaked set $L^{\Gamma}_i$ corresponding to the current constraint system $\Gamma$, and consider all possible ways of expanding it to a larger set $L$: for each exponent variable $e$ that appears in $t$, $h_1$,\ldots, $h_n$, if it is not already in the current leaked set, we case split on whether the adversary could learn $e$ or not. In the first case, we add the constraint $K(e)@j_e, j_e < i$ to $\Gamma$ and add $e$ to $L$. We will explain how we solve these variable constraints in the following subsection. In the second case, we consider $e$ to be secret and we raise a contradiction if $e$ is ever learned by the adversary.  

For each option of $L$, let $\sigma$ be the generalized (via the $gen$ function) unifier of the indicators, with respect to $L$, of $t$ with root terms of terms $h_1,\ldots,h_n$. Then we introduce additional variables $X_1,\ldots,X_n, X_{n+1}$ of sort $\mathit{varE}$, introduce the constraints $K(X_i)@j_i, j_i < i$, and consider the linear equation 
$$\toPoly(X_1\sigma(h_1) + \ldots + X_n\sigma(h_n) +X_{n+1})= \toPoly(\sigma(t)).$$
The variables $X_k$ represent the adversary exponentiating $h_k$ with some terms he can derive and extra variable $X_{n+1}$ represents the adversary multiplying the $h_i$'s with another derivable term.
We match the coefficients of the above polynomials in the variables $S=\{s_1,\ldots,s_l\}$, and we obtain a system of linear equations to solve. Crucially, for every solution, $X_i\in \mathbb{Q}(l_1,\ldots,l_n)$. So by the definition of $L$, $X_i$ is indeed deducible by the adversary, since the adversary can apply all $\mathit{DH}$ operators. The solution hence represents a valid attack.

In general, we try to solve the variable constraints with fresh variables as soon as they are added to the system. As a result, in practice most of the possible $L$'s are directly excluded before proceeding to solve the corresponding linear equations.

\subsubsection{Variable terms}

We now turn to goals of the form $K(v)@i$, where $v$ is a variable of sort $FrE$. For example, the goals we have introduced above are of this form. If $v$ is already in the leaked set $L^{\Gamma}_i$ of the constraint system, we mark this goal as solved. Otherwise, we consider it to be secret, and we try solve the goal with the strategy above. That is, we try to unify this variable in the simplified $\mathit{DH}$ theory with root terms $h_i$ of $Out$ facts containing an argument of sort $E$. If there are other variables appearing in the $h_i$ terms that are not $v$ and that are not in the leaked set $L^{\Gamma}_i$, we case split on them being known or not. Thus the only difference from the previous case is that we do not make a case split on the variable $v$ itself.

For variables of sort $E$, we delay solving $K^{\uparrow}(\mathit{v:E})$ goals until we have solved all other goals and hence collected all the necessary restrictions on the variable $\mathit{v:E}$. Solving other goals often involves a substitution that maps $\mathit{v:E}$ to a different composite term. 

We provide the full details of the constraint solving rules corresponding to adversarial action facts in Appendix \ref{sec:KINpremises}, see Rules \ref{rul:CPremK} and \ref{rul:CvarK}.

\section{Combination with other equational theories. }
\label{sec:combination}
To profit from Tamarin's full expressive power, our approach supports combining the full Diffie-Hellman group structure with other equational theories, both Tamarin's built-in ones (e.g. pairs, symmetric/asymmetric encryption etc.) and user defined ones. In particular, since $\mathit{DH}$ is a subsort of $\mathit{Msg}$, Diffie-Hellman terms are valid arguments of other equational theories' operators. However, as is also natural, we cannot apply $\mathit{DH}$ operators to terms that are not of sort $\mathit{DH}$.

To handle facts that contain combined terms (i.e. terms that contain both $\mathit{DH}$ and non-$\mathit{DH}$ variables), we follow Baader and Schulz's approach \cite{baader-schulz} of combining unification procedures. In particular, we successively replace $\mathit{DH}$ 
sub-terms by new $\mathit{DH}$ variables, effectively `hiding' the DH term. We call this process \emph{cleaning}. Note that we can unify cleaned terms with the user-defined equational theory $\EqUsr$, excluding any $\mathit{DH}$ equation. The solution to this unification problem might return substitutions that unify two different $\mathit{DH}$ variables. If we keep track of what each of these variables replaces, this corresponds to a unification constraint on two now purely-DH terms.  We can then use the approach described in the previous section for these constraints. 

More formally, in Appendix \ref{app:DHdefs}, 
we define a function $\cl$ that \emph{cleans} combined terms by replacing $\mathit{DH}$-terms by fresh variables and storing this information. 
For example, when trying to match the $In(\langle g^{c_1}, g^{c_2} \rangle)$ fact with an $Out(\langle g^y, g^m\cdot pka^y \rangle)$ fact in Example 1, Tamarin cleans the terms by replacing the $\mathit{DH}$ terms with fresh variables: $u_1\mapsto g^{c_1}, u_2\mapsto g^{c_2}, v_1\mapsto g^y,$ and $v_2\mapsto g^m\cdot pka^y$. It then produces the unification query $\unify(\langle u_1,u_2\rangle, \langle v_1, v_2 \rangle)$. This can be solved with Tamarin's current variant-based unification algorithm, returning a most general unifier, $\sigma: \{u_1\mapsto v_1, u_2\mapsto v_2\}$. Tamarin will then use the stored $\mathit{DH}$ terms to transform this unifier into new constraints
$$g^{c_1}=_{DH} g^y\quad\text{ and }\quad g^{c_2}=_{DH} g^m\cdot pka^y.$$
These can be solved using the rules from Section \ref{sec:newrules}. Appendix \ref{app:combrule} defines the corresponding constraint solving rule, and shows that the constraint solving algorithm remains sound and complete.

\section{Implementation and case studies}

We implemented the above strategy in the Tamarin prover, providing a new \texttt{DH-multiplication} built-in theory that enables the use of all operators from  $\Sigma_{DH}$. Our goal was to make the integration as complete as possible, keeping all of the tool's features. This required not only modifying the constraint solving rules, but also all other features that rely on unification. For example, Tamarin allows for \emph{restrictions} that exclude certain protocol executions, and such restrictions may contain equality constraints. We adapted all such equality checks to use our approach for $\mathit{DH}$ terms.  Other features that are less commonly used (e.g. Tamarin's \emph{diff-mode} for verifying observational equivalence properties \cite{obseq}) are left for future work. We have only implemented the Gaussian elimination algorithm within Tamarin, so our implementation only covers protocols that produce linear equations. This excludes for example, the ElGamal signature scheme. Adding support for other algebraic decision procedures (e.g. Buchberger's algorithm) that solve the nonlinear equations generated by such protocols is an implementation task that is left for future work. 

We first evaluate our implementation on existing models from the Tamarin repository \cite{tamarin-repo} that use the limited Diffie-Hellman built-in equational theory. 
Afterwards, to illustrate the effectiveness of our approach on protocols using the group multiplication, we modelled the El Gamal encryption scheme \cite{ElGamal} and the MQV key exchange \cite{mqv2}. 
All the above case studies were performed on a laptop with an Intel i7 6-core processor and 16GB memory.

\subsection{Evaluation}

We evaluate our implementation's correctness using the following models taken from the Tamarin repository: 
\begin{enumerate} 
\item toy protocols using basic adversary capabilities regarding exponentiation;
\item 12 real world protocols that use the basic Diffie-Hellman key exchange in different variants, from 
a benchmark of 20 case studies from \cite{tableDH}. These protocols use different built-in equational theories (symmetric-encryption schemes, hash functions, signature schemes) and different user-defined functions. Each protocol model verifies one or two properties.
\item The Tamarin model of the Wireguard protocol from \cite{wireguardTamarin}, a large case study containing multiple nested hashes, different user-defined functions, and reuse lemmas. This model contains 8 lemmas in total: 2 executability lemmas, 2 agreement properties, 2 secrecy properties, and 2 reuse lemmas to help prove the agreement properties. 
\end{enumerate}
As mentioned, these models do not make explicit use of the group multiplication in the protocol rules. 

While we do allow the combination of different built-in theories, we do not allow different built-in function symbols or user-defined function symbols to be used as \emph{arguments} of our Diffie-Hellman operators. For example, a term of the form $g^{h(e)}$, where $h:\mathit{Msg} \rightarrow Msg$ is a user-defined function symbol, is \emph{not} allowed, since our Diffie-Hellman operator take as argument terms of sort $G$ or $E$.
 
Out of the 20 case studies in the benchmark in 2), two studies used user-defined function symbols as exponents and were hence excluded. Furthermore, we also excluded some models that did not use the Diffie-Hellman built-in theory at all, one model where the proof-search was reported to not terminate, and another model whose source file was missing from the Tamarin repository. 

Enabling users to define function symbols that produce outputs of sort $\mathit{DH}$ is left for future work. In cases where such functions have no associated equations, their integration is expected to be relatively straightforward. However, when these function symbols are subject to user-defined equations, the situation becomes more complex. It would be necessary to ensure that augmenting the current rewrite system $\rightarrow_{\mathit{DH}}$ with these additional equations preserves termination and confluence. Moreover, how to incorporate such equations within the algebraic framework is an open question. 

We rewrote each of these models in our extension's new syntax, replacing the Diffie-Hellman function symbols with the symbols of our new \texttt{DH-multiplication} built-in theory. We did not change anything else, including how the protocol rules are defined, or how the properties are formulated. 
While the protocol itself is unchanged, our extension models a more powerful adversary, capable of using additional operators, namely addition and subtraction. 
Nevertheless, we argue that this enhanced adversary cannot derive new attacks if the protocol rules themselves do not make use of these additional operators.

To give an intuitive argument, observe that any lemma in the protocol necessarily refers only to terms appearing in the protocol’s rules, and thus these terms do not contain the new operators. Consider an an attack trace in which the adversary uses the new operators. Since the attack trace's final goal is to contradict the property, we consider the last adversary rule in the trace that applies one of these new operators. By construction, this rule must produce a term free of these operators in order to match the property. Now consider the structure of such a rule: it must eliminate new operators. For example, one such rule is when the adversary applies the minus operator: 
$$\verb![K(a+b), K(b)] --> [K(a)].!$$ 
In this case, the adversary can only construct 
 $a+b$ if they already know both $a$ and $b$ since the protocol participants themselves cannot have produced $a+b$. 
Thus, the application of the operator-based rule is not essential: if the adversary already knew $a$ then the use of the new operator in this step can be bypassed entirely. 
One can apply such an argument to all such rules in the trace. Therefore, for every attack trace that uses the new operators, there exists a corresponding trace that does not use them. 

Indeed, for every model and its properties, we find the same results as with the restricted Diffie-Hellman theory. Table \ref{tab:results} reports the time it took to prove or disprove all of the properties of each given model. All properties of all models are proved automatically with the restricted Diffie-Hellman theory. Using our extension, for the WireGuard model we manually guided the proof search tree to reduce runtimes. In particular, for WireGuard we manually proved 6 out of the 8 lemmas: we proved the 2 executability lemmas, the 2 reuse lemmas and the 2 secrecy lemmas manually. For this model, the time in Table \ref{tab:results} indicates the seconds or minutes it took to prove all the lemmas with the search strategy stored. 

\begin{table}[ht]
\centering
\begin{tabular}{
    >{\centering\arraybackslash}m{0.1cm} 
    >{\arraybackslash}m{1.8cm}            
    >{\arraybackslash}m{1.4cm}            
    >{\centering\arraybackslash}m{0.75cm}  
    >{\centering\arraybackslash}m{1.2cm}  
    >{\centering\arraybackslash}m{0.75cm}  
}
\toprule
\textbf{\#} & \textbf{Protocol} & \textbf{Result} & \textbf{Time} & \textbf{Time Extension} & \textbf{Auto Prove} \\
\midrule\midrule
\multirow{2}{*}{1} & Toy1 & Proof & 0.1\,s & 0.1\,s & \ding{51} \\
                   & Toy2 & Proof+Attack & 0.1\,s & 0.2\,s & \ding{51} \\
\addlinespace
\multirow{13}{*}{2} 
                   & KEA+ & Proof & 1\,s & 5\,s & \ding{51} \\
                   & KEA+ (wPFS) & Attack & 1\,s & 20\,s & \ding{51} \\
                   & SIG-DH & Proof & 0.4\,s & 2\,s & \ding{51} \\
                   & STS-MAC & Attack & 3\,s & 60\,s & \ding{51} \\
                   & STS-MAC-fix1 & Proof & 9\,s & 4\,m & \ding{51} \\
                   & STS-MAC-fix2 & Proof & 2\,s & 30\,s & \ding{51} \\
                   & TS1-2004 & Attack & 0.2\,s & 1\,s & \ding{51} \\
                   & TS1-2008 & Proof & 0.2\,s & 1\,s & \ding{51} \\
                   & TS2-2004 & Attack & 0.2\,s & 5\,s & \ding{51} \\
                   & TS2-2008 & Proof & 0.4\,s & 5\,s & \ding{51} \\                   
                   & UM (PFS) & Attack & 0.4\,s & 3\,s & \ding{51} \\
                   & UM (wPFS) & Proof & 0.7\,s & 5\,s & \ding{51} \\
\addlinespace
\multirow{1}{*}{3} & WireGuard & Proof & 2\,m & 10\,m & \ding{55} \\
\bottomrule
\end{tabular}
\vspace{0.2cm}
\caption{Comparison of Runtimes}
\label{tab:results}
\end{table}

Our extension results in slower performance. This is to be expected since with new operators, there are also more potential attack traces that Tamarin must consider and exclude.  In particular, verifying the WireGuard protocol takes around 5 times as long as when verified without the Diffie-Hellman extension, even when the proof is manually guided.
We believe that the primary cause of this major slowdown is not precomputing sources for each premise fact, as is instead done in Tamarin without the extension. As a result,  in the current implementation, these sources must be recomputed for every occurrence of a premise. 

More generally, however, the use of the extension is not needed for protocols that do not explicitly rely on group multiplication. In the next sections, we therefore look at protocols that do use this operation. 

\subsection{ElGamal}
In Section \ref{sec:overview}, we provided a Tamarin model of the ElGamal encryption scheme. We defined an executability lemma and a lemma about the secrecy of the encrypted message from Bob's perspective. We also add a lemma about the secrecy of the received message from Alice's perspective:
\begin{lstlisting}
lemma secrecyA:
"All msg #i A. SecretA(A, msg)@i 
              & not (Ex #l. Compromised(A)@l) ==> 
       not (Ex #j. K(msg)@j )"
\end{lstlisting}
 
Tamarin automatically proves the executability lemma, providing a valid trace of the protocol, which takes three minutes. With manual guidance in its search tree, Tamarin proves the lemma in a few seconds. Tamarin also proves secrecy of the encrypted message from Bob's perspective, automatically, also in a few seconds. 

For the lemma \texttt{secrecyA}, since the non-cancellation property on $c_1^{-ska}\cdot c_2$ does not hold, a priori there are no security guarantees if no counterexample is found. However, in this case, Tamarin automatically disproves the property, finding an attack in a few seconds. This is not surprising, since Bob is not authenticated and an adversary can send Alice any message of his choice. Thus the approach can be useful even if the non-cancellation property fails, as the counterexamples found are always valid attacks. 

ElGamal's security has been extensively studied and it is widely regarded as secure, based on the hardness assumptions of the discrete logarithm problem. While proving its security in the symbolic model may not seem to be groundbreaking, we believe it represents a significant advance: in spite of ElGamal's widespread use, this is the first time it is automatically verified, in the symbolic model. 

\subsection{MQV Key Exchange}
\label{sec:exmqv}
We now turn to a more complex example. MQV \cite{mqv, mqv2} is an authenticated key exchange protocol where the key is computed by combining static and ephemeral key pairs. The shared key should not agree if the other party’s public key is not employed, thus giving implicit authentication of the correct parties. MQV's main attractiveness is its computational efficiency; it claims to achieve authentication in two message exchanges and a few group operations. In \cite{mqv-attack}, Kaliski presented an unknown key-share attack on this protocol, where Eve can trick Bob into thinking he is sharing Alice's key with Eve instead of with Alice.  

The original protocol is as follows. Alice (in the \emph{initiator} role) and Bob (in the \emph{receiver} role) both have a long-term private/public key pair $\langle a,pk_A=g^a\rangle$ and $\langle b,pk_B=g^b \rangle$ and they each compute ephemeral key pairs $\langle x,X=g^x\rangle$ and $\langle y,Y=g^y\rangle$.  
Alice and Bob send each other their ephemeral public keys $X$ and $Y$ and it is assumed they know each other's public key. Alice and Bob compute

\begin{equation}
K_{AB} = (Y\cdot pk_B^{\mu(Y)}) ^{x+\mu(X)a}, 
\label{eq:KAB}
\end{equation}
\begin{equation}
K_{BA} = (X\cdot pk_A^{\mu(X)}) ^{y+\mu(Y)b}.
\label{eq:KBA}
\end{equation}
By the algebraic properties of the operators, $K_{AB}=K_{BA}$. 

Kaliski's unknown key-share attack \cite{mqv-attack} is as follows. Alice initiates a protocol run with Bob. 
\begin{enumerate} 
\item[i)] Eve first intercepts Alice's ephemeral key $g^x$. She selects a random value $r_E$ and computes the ephemeral public key $$Z =g^x\cdot (g^{a})^{\mu(g^x)}\cdot g^{r_E}.$$ Observe that Eve does not know the secret key corresponding to this public key. 
\item[ii)] Eve computes the following private, public key pair: 
$$ \langle e = -\frac{r_E}{\mu(Z)}, pk_E = g^{e} \rangle.$$
\item[iii)] Eve initiates a protocol run with Bob as herself, i.e. with public key $pk_E$ from (ii) and using the ephemeral key $Z$ computed in (i).
\item[(iv)] Eve receives Bob's key and forwards it to Alice.  
\end{enumerate}
After these steps, Alice will compute the shared key $K_{AB}$ as before, while Bob will also compute the same key $K_{BE}$ but with \emph{Eve}:
\begin{align*}
K_{BE}  &= (Z \cdot pk_E^{\mu(Z)}) ^{y+\mu(Y)b} \\
 & = (g^x\cdot g^{a \mu(g^x)} \cdot g^{r_E}\cdot g^{-r_E}) ^{y+\mu(Y)b} \\
 & =  (X\cdot pk_A^{\mu(X)}) ^{y+\mu(Y)b} = K_{AB}.
\end{align*}
This is the same key obtained in an honest key-share with party $A$, see (\ref{eq:KBA}).
Observe that Eve does, in fact, not know this key. However there is still a mismatch between the identities who share the key. Kaliski names such an attack an \emph{unknown key-share} attack. 
Observe that this attack is only possible on the responder: Eve needs the initiator's ephemeral key to generate her malicious ephemeral key, making it impossible to also trick the initiator. 

Modelling the MQV protocol in Tamarin equipped with the new $\mathit{DH}$ equational theory is straightforward. We provide the Tamarin model in Appendix \ref{sec:mqvmodel}. We consider an adversary that can obtain certificates for arbitrarily chosen private/public key pairs. Observe that, without a subsequent message exchange using the shared key, there can be no agreement on the parties' identities: Alice and Bob just compute a shared key with the (possibly adversarial) ephemeral public key they receive. The authentication is implicitly obtained when decrypting/encrypting messages with a key that matches the one they have computed. Indeed, our extension of Tamarin quickly (in less than a second) and automatically falsifies the agreement properties. 

We thus expand the protocol model to include two messages where Bob and Alice exchange fresh nonces encrypted with the shared key (see Appendix \ref{sec:mqvmodel}).  In particular we use Tamarin's built-in symmetric encryption primitive.  

Below, we state two basic security properties regarding the secrecy of the key and agreement on the key and parties sharing the key. 
\begin{lstlisting}
lemma secrecyI:
 "All I R key #i. AgreeKeyI(I, R, key)@i & (not (Ex #k. Compromised(I)@k))& 
 (not (Ex #k. Compromised(R)@k))
  ==> not (Ex #i. K(key)@i)"
  
lemma agreementI:
"All I R key #i. AgreeKeyI(I, R, key)@i & (not (Ex #k. Compromised(I)@k))& 
(not (Ex #k. Compromised(R)@k)) 
 ==> Ex #j. RunningR(R, I, key)@j"
\end{lstlisting}



The action fact $AgreeKeyI$ is declared in the rule where Alice correctly decrypts a message from Bob using the shared key she computed. $RunningI$ is declared in the first rule where Alice can compute the shared key. We also declare $AgreeKeyR$ and $RunningR$ action facts in Bob's respective rules and we define similar \texttt{secrecyR} and \texttt{agreementR} properties from Bob's perspective. Finally, we also add an exists-trace lemma \texttt{executable} that verifies if Bob and Alice can actually complete all their steps, exchanging a key. Note that the keys satisfy the non-cancellation property by Theorem \ref{thm:NoCanc} of Appendix \ref{sec:appendixdefs}.

Observe that Kaliski's unknown key-share attack is a counterexample to \texttt{agreementI}, since in that scenario Alice agrees to share her key with Bob, while Bob does not think he is sharing his key with her. 

We manually guide Tamarin's proof for the executability lemma to find the expected execution. 

We first consider the \texttt{secrecyI} lemma and illustrate how Tamarin handles the $\NoCanc$ property. In particular, this lemma   refers to the key that Alice computes: 
\begin{align*}
\mathit{key} & = (g^y \cdot g^{b\mu(g^y)})^{x + a\mu(g^x)} \\
& = g^{y x}\cdot g^{ya\mu(g^x)}g^{bx\mu(g^y)}g^{ba\mu(g^y)\mu(g^x)}.
\end{align*}
Recall that this key corresponds to Alice's view, so only the variables $a$ and $x$ are of sort $\mathit{FrE}$: $a$ is Alice's secret key and hence stored as a fresh variable in her state facts, and $x$ is freshly generated by herself. The variables $y$ and $b$, corresponding to Bob's ephemeral and private key, are instead $E$ variables, since they are possibly adversarial. Hence, at this point, $\NoCanc(\mathit{key})$ does \emph{not} hold. However, our lemma excludes the case where Bob is compromised since trivially the key is \emph{not} secret if we share it with the adversary. Hence, Tamarin can conclude that $b$ comes from an honest private/public key generation fact, meaning that $b$ must be of sort \textit{FrNZE}. At this point $\NoCanc(\mathit{key})$ still does not hold, since $y$ can be adversarially chosen. However, by Theorem \ref{thm:NoCanc},
 $$\NoCanc'(g^{ba\mu(g^y)\mu(g^x)}, \mathit{key})$$ does hold.
 When trying to unify indicators, Tamarin will prove that the indicator of $g^{ba\mu(g^y)\mu(g^x)} = g^{ba}$ cannot be found in any root term of any $Out$ fact. Since $g^{ba\mu(g^y)\mu(g^x)}$ does satisfy $\NoCanc$ with the other terms, this is enough to soundly conclude that the entire term cannot be constructed. We explain the constraint solving process for this lemma in Appendix \ref{sec:mqvmodel}. 
The secrecy lemmas are proven automatically in a few seconds, since the infeasibility of reaching the necessary indicator terms is established quickly. 

 
As for the property \texttt{agreementI}, Tamarin returns different attacks. The attacks are valid unconditionally of the $\NoCanc$ property, which in fact does \emph{not} generally hold on the root terms of the keys shared with compromised parties. One attack is the trace illustrated in Figure \ref{fig:mqvattack1}, where a third party (E) can perform an unknown key-share attack using a trivial public key.
\begin{figure}[h]
\includegraphics[width=9cm]{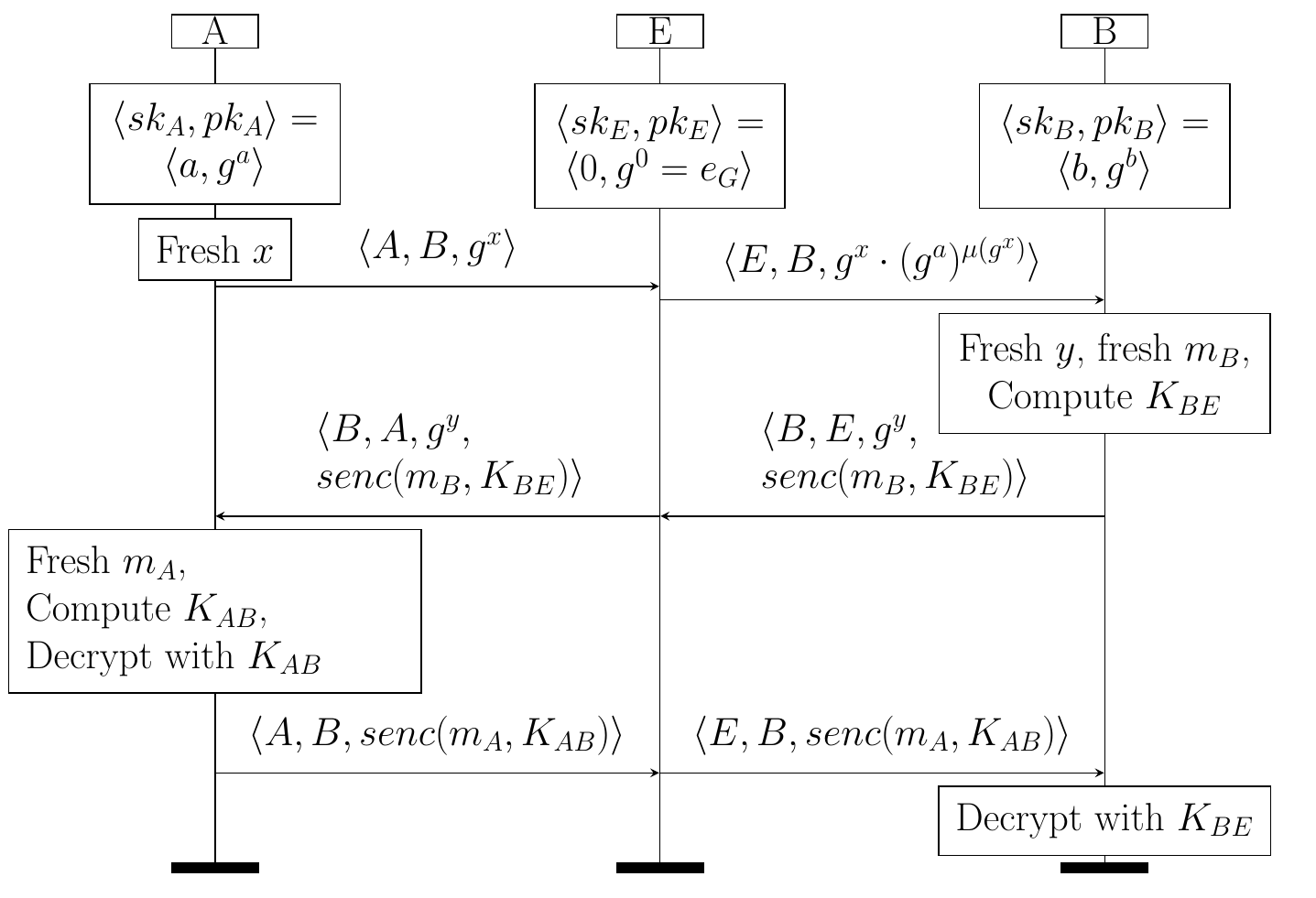}
\caption{Attack trace returned by Tamarin.}
\label{fig:mqvattack1}
\end{figure} 
This attack is essentially an instantiation of Kaliski's unknown key-share attack, where Eve's random value is replaced with $r_E = 0$.
While we again do \emph{not} have agreement on the parties' identities, Eve still does not actually learn the key since Alice's and Bob's secret keys ($x$ and $y$) are not compromised.
Even if this attack is unrealistic in practice (assuming participants check that public keys are not trivial), the chosen ephemeral public key Eve must use, $Y\cdot A^{\mu(Y)}$, involves group multiplication and we believe this attack is not straightforward to find by hand. 

If we exclude ephemeral and public keys from being trivial group elements, Tamarin also finds the attack proposed by Kaliski.  

A crucial step in Tamarin’s constraint-solving algorithm for finding the attack is the unification of the session keys in Alice’s and Bob’s rules. Alice’s key is $$g^{y_Ex}\cdot g^{b\mu(g^{y_E})x}\cdot g^{y_E\mu(g^x)a} \cdot g^{b\mu(g^{y_E})\mu(g^x)a},$$ while Bob’s key is $$g^{x_E y}\cdot g^{e_E \mu(g^{x_E})y} \cdot g^{x_E\mu(g^{y})b}\cdot g^{e_E\mu(g^{x_E})\mu(g^{y})b}$$
Here, $g^a$ and $g^x$ are Alice’s long-term and ephemeral keys, and $g^b$ and $g^y$ are Bob’s corresponding keys. The terms $g^{e_E}$, $g^{x_E}$, and $g^{y_E}$ contain unbound variables of sort $E$.

Tamarin proceeds by unifying the root terms of Alice’s key with (possibly repeated) root terms of Bob’s key. Kaliski’s attack arises when the root terms $g^{y_E x}$ and $g^{y_E\mu(g^x)a}$ are both matched to $g^{x_E y}$ by splitting $x_E$ into a sum, and the root terms $g^{b\mu(g^{y_E})x}$ and $g^{b\mu(g^{y_E})\mu(g^x)a}$ are both matched to $g^{x_E\mu(g^{y})b}$. The variables introduced by the generalized unifier are then instantiated with algebraic methods to make the other root terms of Bob's key cancel out.
Appendix \ref{sec:mqvmodel} explains the rules used to find the attack in more detail.

Finding this attack required manual interaction. Re-running Tamarin with the solving strategy stored, it takes around half a minute for Tamarin to return the attack. 

Tamarin instead excludes attacks from Bob's (the responder) perspective, proving \texttt{agreementR}. As for proving the secrecy lemmas, proving this lemma also required manual interaction. When re-running the proof with the solving strategy stored, it takes Tamarin again around a half a minute to verify \texttt{agreementR}. Writing specific heuristics for DH extended theories (deciding which constraint solving rules to apply first) would probably be enough for Tamarin to find the proofs and attacks automatically in a reasonable amount of time. 
 
We therefore believe our approach scales to fairly complex protocols (e.g. using different equational theories and Tamarin features) albeit with some user interaction. Other protocols, such as HMQV (the improved version of MQV) use different function symbols as exponents. In particular, HMQV uses a hash function of arity 2 as an exponent. We hence plan to integrate arbitrary user-defined function symbols \emph{with no associated equations} that produce outputs of sort $\mathit{DH}$ to perform further large case studies. 

\section{Conclusion}

We have presented a method for the symbolic analysis of protocols utilizing all Diffie-Hellman group operators, including addition in the exponents. We have successfully integrated this method into the Tamarin prover, providing the first symbolic verifier that supports such protocols. Our integration maintains compatibility with Tamarin’s basic features, built-in equational theories and user-defined ones.
We have verified the effectiveness of our approach by analyzing the ElGamal and MQV protocol. 

To further improve our approach, we plan to write heuristics to make Tamarin's search more automatic, to enable user-defined $\mathit{DH}$ function symbols and to perform case studies on other large protocols such as HMQV \cite{HMQV}.

\bibliographystyle{IEEEtran}
\bibliography{draft}

\appendices

\section{Definitions omitted in main text}
\label{sec:appendixdefs}
In Sections \ref{sec:notation} and \ref{sec:depgraphs}, we recall the formal definition of dependency graphs used by the Tamarin prover, and how we modify them slightly to incorporate our equational theory. In particular, we show that we can use Tamarin's constraint solving approach to find counterexamples to security properties. 
\subsection{Notational Preliminaries}
\label{sec:notation}
We present here the notation used in \cite{Tamarin} and \cite{Benedikt}. 

For a sequence $s$, we write $|s|$ for the length of $s$ and $id_x(s)=\{1,\ldots,|s|\}$ for the set of indices of $s$. $S^{\#}$ denotes the set of finite multisets with elements from $S$. 

For an equational theory $\mathcal{E}$, we define the $\E$-instances of a term $t$ as $inst_\varepsilon=\{t' | \exists \sigma. t\sigma =_\varepsilon t'\}$, and we call the ground $\E$-instances $ginsts_\E(t)$.

Recall that we consider a user defined signature $\Sigma_{Usr}$, that we assume disjoint from $\Sigma_{DH}$, with an equational theory $\mathcal{E}_{Usr}$. We assume this theory can be decomposed (by orienting equations) in a rewrite system $R_{Usr}$. Finally, if $\Sigma_{DH}$ contains $AC$ operators, it is assumed that $(\Sigma_{Usr}, R_{Usr}, AC)$ has the finite variant property. 

\begin{definition}
Let $\mathcal{E}_{DH}$ be the set of unoriented equations over $\Sigma_{DH}$ from Definition \ref{def:dheq}.
\end{definition}

We have defined $\mathcal{T}'$ to be the term algebra over $\Sigma$, $\mathit{PN}$, $\mathit{FN}$, $FN_G$, $\mathit{PN}_G$, $\mathit{FN}_{\mathit{E}}$, $PN_{\mathit{E}}$, $\mathcal{V}$, $\mathcal{V}_G$, and $\mathcal{V}_E$, where
 \begin{align*}
\Sigma & :=\Sigma_{Usr}\cup \Sigma_{DH},\\
\mathcal{E} &:=  \mathcal{E}_{Usr} \cup \mathcal{E}_{DH}.
\end{align*}

Finally, we also assume an unsorted signature $\Sigma_F$ partitioned into linear and persistent fact symbols. 
We say that a fact $F(t_1, ,t_k)$ is linear if $F$ is linear and persistent if $F$ is
persistent. Linear facts model resources that can only be consumed once, whereas persistent
facts model resources that can be consumed arbitrarily often.

A protocol is modelled as labelled transition systems: protocol rules and adversarial capabilities are represented as multiset rewrite rules. Given a set $R$ of multiset rewrite rules, $ginsts(R)$ denotes the ground instances of $R$.

Recall that for a rule \texttt{ru = l-[a]->r}, we define the premises as $prems(ru) = l$, the actions as $acts(ru) = a$, and the conclusions as $concs(ru) = r$. 

Finally, as mentioned in Section \ref{sec:newsorts}, we define following set of adversarial message deduction rules.
\begin{definition}
We define $\mathit{MD}_{\mathit{DH}}$ to be the set of the following message deduction rules:
\begin{lstlisting}
rule pubG: []--[K(g:PubG)]-> [K(g:PubG)]
rule FrE: [Fr(x:FrE)]--[K(x:FrE)]->[K(x:FrE)]
rule zero: []--[K(0)]->[K(0)]
rule one: []--[K(1)]->[K(1)]
rule eG: []--[K(eG)]->[K(eG)]
rule mu: [K(t)] --[K(mu(t))] -> [K(mu(t))]
\end{lstlisting}

Furthermore, for each operator $f\in \Sigma_{DH}$ of arity $n$, we add a rule allowing the adversary to apply it
\begin{lstlisting}
rule f: [K(x1),...,K(xn)]
        --[K(f(x1,...,xn))]->[K(f(x1,\ldots,xn))] 
\end{lstlisting}
\end{definition}

In particular, we distinguish these rules from Tamarin's built-in adverarial rules $MD$. 

\subsection{Dependency graphs}
\label{sec:depgraphs}
We modify slightly Definition 3.9 of \cite{Benedikt}, to consider equality modulo $\mathcal{E}_{DH}$ instead of syntactic equality. 
\begin{definition}
\label{def:dgbasic}
Let $P$ be a set of labeled multiset rewriting rules, with terms from $\mathcal{T}'$.
A $\mathit{DH}$-extended dependency graph modulo $\E_{Usr}$ for $P$ is a pair $dg=(I,D)$ where $I\in ginsts_{\E_{Usr}}(P)$, $D\in \mathbb{N}^2 \times \mathbb{N}^2$ and $dg$ satisfies the conditions below. 

The nodes of $dg$ are defined as $idx(I)$ and $D$ are the edges of $dg$. $((i,u), (j,v))$ is abbreviated by $(i,u)\rightarrow (j,v)$.
A conclusion of $dg$ is a pair $(i, u)$ such that $i$ is a node of $dg$ and $u\in idx(concs(I_i))$. The corresponding conclusion fact is $(concs(I_i))_u$. A premise of $dg$ is a pair $(i,u)$ such that $i$ is a node of $dg$ and $u\in $idx(prems(I$i))$. The corresponding premise fact is $(prems(I_i))_u$. A conclusion or premise is linear if its fact is linear.

\begin{enumerate}
\item For every edge $(i,u)\mapsto(j,v)\in D$, it holds that $i<j$ and the conclusion fact of $(i,u)$ is equal to the premise fact of $(j,v)$, modulo $\mathcal{E}_{DH}$.
\item Every premise of $dg$ has exactly one incoming edge.
\item Every linear conclusion of $dg$ has at most one outgoing edge.
\item The instances of rules creating $\mathit{Fresh}$ or $\mathit{FrE}$ variables are unique.
\end{enumerate}
For a protocol $P$ we denote the set of all possible such graphs by $dgraphs_{DH}^{\E_{Usr}}(P).$ 
\end{definition}

Let $I= \{l_1- [a_1]\rightarrow r_1,\ldots, l_k - [a_k]\rightarrow r_k\}$, we define the trace of $dg$ as $trace(dg) = [a_1,\ldots,a_k]$

For all protocols (i.e. set of multiset rewrite rules) $P$ over $\mathcal{T}'$ and message deduction rules $D$, it follows that:
$$trace(execs(P\cup D)) = trace(dgraphs_{DH}^{\E_{Usr}}(P\cup D))$$
This is equivalent to Lemma 3.10 in \cite{Benedikt}, for $\mathit{DH}$-extended dependency graphs.

As done in \cite{Benedikt}, we extend the previous definition to $\mathit{DH}$-extended dependency graphs modulo associativity and
commutativity (AC).

\begin{definition}
Let $R_{DH}$ be the set of $\mathit{DH}$ rewrite rules from Definition \ref{def:dheq}, which are coeherent and convergent. We define $\mathcal{R}= R_{DH} \cup R_{Usr}$. 
\end{definition}

Since the two rewrite systems have disjoint signatures, we can consider the rewrite systems individually. They are both convergent and coherent, thus there is a unique normal form with respect to $\R,AC$ rewriting, denoted by $t\downarrow_\R$.

Given a term $t$, denote by $\lceil t \rceil ^{var}$ the complete set of $R_{Usr},AC$-variants. Denote by $\lceil t \rceil ^{R_{Usr}}_{insts}$ the set $$\{\tau(t)\downarrow_R| \tau\in \lceil t \rceil ^{var} \}.$$ Finally, a term $t$ is $\downarrow_\R$-normal, if $t=_{AC}t\downarrow_{\R}$. We extend this to dependency graphs, $dg=(I,D)$ is $\downarrow_\R$-normal is all rule instances in $I$ are. 
Similarly to Lemma 3.11 of \cite{Benedikt} it follows that:
\begin{lemma}
For all set of multiset rewrite rules $R$,
\begin{align*}
&dgraphs_{DH}^{\E_{Usr}}(R)\downarrow_\R = \\
&\{ dg\ |\ dg\in dgraphs_{DH}^{AC}(\lceil t \rceil ^{R_{Usr}}_{insts}), dg \downarrow_\R normal \}.
\end{align*}

\end{lemma}

Finally \cite{Benedikt} extends the notion of dependency graphs to consider normalized instances of Tamarin's messaged deduction rules $MD$. In particular, the Schmidt define normal message deduction rules $ND$, and defines normal-form conditions for dependency graphs where the rules $MD$ are replaced by $ND$. 
Considering our $\mathit{DH}$-message deduction rules $MD_{DH}$ as normal protocol rules, we can use the same normal-form conditions. In particular, we consider the graphs in
$$dgraphs_{DH}(\lceil P\cup MD_{DH}\rceil^{(R_{Usr})}_{insts} \cup ND )$$ 

and impose the same normal-form conditions as Definition 3.18 of \cite{Benedikt}. We call the set of all these graphs for $P$ as $ndgraphs_{DH}(P)$. 
Finally, as in Lemma 3.19 of \cite{Benedikt}, we obtain that for all protocols $P$, 
\begin{align*}
&trace(execs(P\cup MD_{DH} \cup MD)) =_{AC} \\
&trace(ndgraphs(P))
\end{align*}

Hence when searching for counterexamples to security properties, we can still search for dependency graphs instead of protocol executions. 

\subsection{Constraint solving}

We use the definition of trace properties from \cite{Benedikt}. We add to their definition of constraints and constraint solving rules, some new constraint definitions and new constraint solving rules (see Appendices \ref{app:newrules} and \ref{sec:KINpremises}). 
As done in \cite{Benedikt}, we define a structure to be a tuple $(dg,\theta)$ where $dg$ is a $\mathit{DH}$-extended dependency graph and $\theta$ a valuation. 
We will use some of the constraints introduced in \cite{Benedikt} in the following sections. 
These constraints and their satisfaction relation $\vdash$ with structures are as follows:
\begin{itemize} 
\item \emph{A node constraint $i:ri$}. $(dg,\theta)\vdash i:ri$ if $\theta(i)\in idx(I)$ and $\theta(ri) = I_{\theta(i)}$ 
\item \emph{An edge constraint $(i,u)\rightarrow(j,v)$}. $(dg,\theta)\vdash (i,u)\rightarrow(j,v)$ if $(\theta(i), u)\rightarrow (\theta(j),v)\in D$
\item \emph{A formula constraint $Fact(t)@i$.} $(dg,\theta)\vdash Fact(t)@i$ if $\theta(i)\in idx(I)$ and $Fact(\theta(t))\in acts(I_{\theta(i)})$
\end{itemize}

For a protocol $P$ and message deduction rules $MD_{DH}\cup ND$, Definition 3.35 of \cite{Benedikt} describes a class of constraint systems that admit a $P$ model.  Essentially, the definition requires that for every edge constraint $(i,u)\rightarrow (j,v)$ there are corresponding node constraint $i:ri$ and $j:rj$ such that premise and conclusion fact matched. Essentially, this definition ensures that all premises are matched by a conclusion, so that a model extraction is straightforward.  
We extend the same definition to including our $MD_{DH}$ rules. 

Theorem 3.36 of \cite{Benedikt} states that every well-formed constraint system $\Gamma$ for $P$ that is solvable with respect to $\rightarrow_P$ has at least one $P$-model and we can extract such a model. Hence the goal of our constraint solving algorithm is still that of constructing a well-formed constraint system. 

\subsection{Definitions regarding $\mathit{DH}$ terms}
\label{app:DHdefs}
\begin{definition}
Let $L$ be a set of exponent variables. We define the indicator $Ind_L(t)$ of a term $t$ as follows: 
\begin{align*}
&Ind_L(X)=\\
&\begin{cases}
exp(\Ind_L(Z),\Ind_L(Y)) & X = exp(Z,Y)\\
\Ind_L(Z) & X = Z^{-1}\\
\Ind_L(Z)& X= -Z \\
\Ind_L(Z)& X= \dhInv(Z)\\
\Ind_L(Z)*\Ind_L(Y) & X=Z *Y \\
1 & X : E \wedge X \in L \\
X  & X : E \wedge X \notin L\\
X & X : G \wedge X \text{ is ground} \\
\end{cases}
\end{align*}
\end{definition}

\begin{definition} We define the function $cl$ that cleans combined terms from their $\mathit{DH}$ parts. 
$$\begin{cases}
\cl(t) = \langle v, \{v\mapsto t\} \rangle,  \ v\in \mathcal{V}_G, \text{  if }t\text{ of sort }G\\
\cl(t) = \langle v, \{v\mapsto t\} \rangle , \ v\in \mathcal{V}_E, \text{  if }t\text{ of sort }E\\
\cl(f(t_1,\ldots,t_n)) = \langle f(\cl_1(t_1),\ldots, \cl_1(t_n)),  \\ 
\quad\quad\quad\quad\quad\quad\quad\ \ \cl_2(t_1)\cup...\cup \cl_2(t_n)\rangle \\
\quad\quad\quad\quad\quad\quad\quad\quad\quad  \text{  if }f(t_1,\ldots,t_n)\text{ not of sort }DH \\
\cl(t)= t \quad \text{otherwise}.
\end{cases}$$
\end{definition}

\noindent In the above $cl_1$ and $cl_2$ indicate respectively the first and second element of the tuple returned by the function $cl$. We assume that the variables introduced in the previous definition are always newly chosen variables. 
For example, $$\cl(\langle g^y, g^m\cdot pka^y \rangle)= \langle \langle v_1, v_2\rangle,\{v_1\mapsto g^y, v_2\mapsto g^m\cdot pka^y \} \rangle.$$

\begin{theorem*}[4 (from Section 4)]
\label{thm:NoCanc}
The condition $\NoCanc_G(X,Y)$ holds for terms $X$ and $Y$ if:
\begin{enumerate}
\item $X =g^{b*\mu(g^t)}$ and $Y=g^{a*t} $, or
\item $X = g^{a*\mu(g^{t_1})}$ and $Y = g^{b*\mu(g^{t_2})}$, 
\end{enumerate}
where $a$ and $b$ are terms that only contain distinct variables of sort $\mathit{FrE}$ and no $\dhInv$ function application.  
\end{theorem*}
Essentially, this follows from the fact that fresh variables cannot be instantiated with function applications. Hence in $1)$ only the $t$ term can be instantiated with an  term containing the $-$ operator, creating an infinite recursion since $t$ appears in both $X$ and $Y$. In $2)$ the $-$ symbol can only appear in a $\mu$ term, which however has no associated equations, and hence cannot be `pulled out'.

\begin{proof}
\begin{enumerate}
\item Let $\theta$ be an arbitrary substitution. Then $\theta(t)$ cannot contain the assitive inverse of $\mu(g^{\theta(t)})$, as this is a recursive definition. Since $a$, $b$ are fresh terms, they cannot be instantiated with the $-$ function either.
Hence $$\theta(a*t)\neq -(\theta({b*\mu(g^t)})),$$ 
which also implies 
$$\theta(g^{a*t})\neq (\theta(g^{b*\mu(g^t)}))^{-1}$$
for all $\theta$ and viceversa.  Finally $\mu(v)\neq 0$ for any term $v:G$, since $\mu$ does not have any associated equations, implying that $g^{b*\mu(g^t)} \neq e_G$ for all $\theta$.
\item Again, let $\theta$ be an arbitrary substitution. By the same reasoning as above,  $g^{a*\mu(g^{t_1})} \neq e_G$ for all $\theta$. 
Next, nor $a$ nor $b$ can be instantiated with a function operator, since they are variables of sort $\mathit{FrE}$. Thus for the terms two cancel out $\mu(t_1)$ and $\mu(t_2)$ must be their additive inverses. However, a $\mu()$ term can never be matched to a $-$ term, since $\mu$ has no associated equations.  Hence again $\theta(a*\mu(g^{t_1}))\neq -(\theta({b*\mu(g^{t_2})}))$ for all $\theta$ and viceversa, implying that for all $\theta$
$$\theta(g^{a*{t_1}})\neq (\theta(g^{b*\mu(g^{t_2})}))^{-1}.$$
\end{enumerate}
\end{proof}

\subsection{Algebraic setting}
\label{app:algdef}
\begin{definition}

We define the function $toQ(\cdot)$ that, given $t:E \in Gen_{DH}({l_1,\ldots,l_n})$, returns an element in $\mathbb{Q}(l_1,\ldots,l_n)$. 
$$toQ(t) = \begin{cases}
toQ(t1) + toQ(t2) & \text{ if } t = t1 + t2 \\
- toQ(t1) & \text{ if } t = - t1 \\
toQ(t1) \cdot toR(t2) & \text{ if } t =t1 * t2 \\
1/toQ(t1) & \text{ if } t = \dhInv(t1) \\
1 & \text{ if } t = 1 \\
0 & \text{ if } t = 0 \\
l_i & \text{ if } t = l_i \\
\end{cases}
$$
\end{definition}

\begin{definition}
Given terms $\{s_1,\ldots,s_l\}$, the function $\toPoly(t)$ interprets a term $t_i\in \Gen_{DH}({l_1,\ldots,l_n,s_1,\ldots,s_l,Y_i})$ as a polynomial with coefficients in $\mathbb{Q}(l_1,\ldots,l_n)[Y_i]$ and variables $s_1,\ldots,s_l$, by naturally interpreting each $s_i$ appearing in $t$ as a variable, and the rest as coefficients in $\mathbb{Q}(l_1,\ldots,l_n)[Y_i]$ via the function $toQ(\cdot)$.
\end{definition}

\section{The new constraint solving rules}
\label{app:newrules}

\subsection{Defining new constraints}
\label{sec:semanticdefs}
We extend Schmidt's original definition of a constraint system \cite{Benedikt}, by introducing new constraints. 
We also define their satisfaction relation with dependency graphs and valuations. 
\begin{itemize}
\item For terms $t$, $h$ and $L$, we define $Eq(t, h)$ to be a particular equality constraint between two terms $t$ and $h$ of sort $\mathit{DH}$. 
$(dg,\theta) \vDash Eq(t, h)$  if and only if $$\theta(t) =_{\varepsilon_{DH}} \theta(h).$$ 
Observe that equality is defined modulo the full Diffie-Hellman theory.
\item We have already introduced the $t\dheq h$ constraints in Section \ref{sec:newrules}. 
$(dg,\theta) \vDash t\dheq h$ if and only if $$\theta(h) =_{\varepsilon_{DH}} \theta(t).$$
\item Given a term $t$ such that $\roots(t)=\{t_1,\ldots,t_n\}$, we introduce the constraint $EqK(t, \langle h _1,\ldots,h_n\rangle,L)$, where $L$ is a set of exponents.
$(dg,\theta) \vDash EqK(t, \langle h _1,\ldots,h_n\rangle, L)$ if and only if for all $i$, $Ind_L(\theta(t_i))=Ind_L(h^{(j)}_i)$, where $h^{(j)}_i$ is a root term of $\theta(h_i)$.
\end{itemize}

We introduce two rules, corresponding to $S_@$ and $S_{Prem}$ respectively. Instead of unifying $Fact(t)=Fact(h)$, (where $Fact$ is not the adversarial $K$ fact), the rules unify indicators of root terms of $t$ and root terms of $h$. 
 
 \begin{rul}[$C_{@}$]
\namedthmlabel{rul:C@} Assuming $\NoCanc(t)$,
$$\crule{Fact(t)@i,\ \Gamma}{\Large\substack{\{Eq(t,h),\ Fact(t)@i,\\ \quad i:ri,\ \Gamma\} _{ri\in \mathcal{P}, Fact(h) \in acts(ri)} }}{C_{@}}.$$
\end{rul}
 
Observe that instead of a standard equality constraint $t=h$, we have introduced a new constraint $Eq(t,h)$. 
We define a similar rule that matches premises and conclusions of rules.

\begin{rul}[$C_{Prem}$]
\namedthmlabel{rul:CPrem} Assuming $\NoCanc(t)$,
$$\crule{ \Large{i:ri,\ Fact(t)\in prems(ri),\ \Gamma }}{\Large\substack{\{ i:ri,\ j:ru,\ Eq(t, h),\\ \qquad \ j<i,\ \Gamma\} _{ru\in \mathcal{P}, Fact(h) \in concs(ru)} }}{C_{Prem}}.$$
\end{rul}


Recall that we denote by $\mathcal{E}_{simp}$ the simplified equational theory that excludes multiplication, addition, and its associated laws directly. Using this theory, we want to check the root terms of $t$ can be unified with root terms of $h$. 

Let $t_1,\ldots,t_n$ be the root terms of $t$ and let $\roots(h)=\{h_1,\ldots,h_m\}$ be the root terms of $h$. 
We construct unification queries by considering combinations of root terms from $\roots(h)$, allowing repetition. Concretely, we select a sequence of $s$ (possibly repeated) elements $\roots(h)$:
$$h_{j_1},\ldots,h_{j_n}.$$ 

The intuition is that if distinct root terms of $t$ (e.g., $t_1$ and $t_2$) originate from the same root term of $h$, then that root term must internally contain an addition. Since the simplified theory $\mathcal{E}_{\mathit{simp}}$ does not introduce additions on its own, we explicitly account for this possibility by introducint substitutions with the $+$ operator.

Formally, suppose that some $h_{j_i}$ appears $p$ times in the chosen sequence, with $2 \leq p \leq n$.  If $h_{j_i}$ does not contain any exponent terms of sort $E$, we leave the sequence unchanged. Otherwise, let $e$ be an exponent variable of sort $E$ occurring in $h_{j_i}$. We introduce fresh variables $f_1,\ldots,f_p$ of sort $E$ and define the substitutions $$\sigma:\{e\mapsto f_1+..+f_P\}, \qquad \{\sigma_i:e\mapsto f_i\}_i.$$

We then construct modified root terms as follows: for each occurrence of $(h_{j_i})_k$, $k=1,...,p$ in the sequence, we replace it by $\sigma_k(h_{j_i})$, and for all other root terms $h_{j_k}$ we apply $\sigma$, obtaining $\sigma(h_{j_k})$.
This yields a new sequence
$$h'_{j_1},\ldots,h'_{j_n}.$$

Let $\mathcal{H}$ denote the set of all sequences of length $n$ obtained by applying the above transformation to all possible choices of repeated root terms. Using these sequences, we define the set of unification problems

\begin{align*}
eqs(t,\roots(h)) = 
\bigl\{
\langle t_1 = h'_{j_1},\ldots, t_n = h'_{j_n} \rangle
\;\big|\;
h'_{j_k} \in \mathcal{H}
\bigr\}.
\end{align*}

Intuitively, each element of $eqs(t,\roots(h))$ attempts to unify every root term of $t$ with a (possibly repeated) root term of $h$. 
or each such equation system $eqs$, Tamarin uses variant-based unification algorithm $\unify_{\Eqsimp}(eqs)$, which return a complete set of most general unifiers. 
In practice, it is often unnecessary to consider all most general unifiers. 
\begin{itemize}
\item If Tamarin returns a substitution $\sigma\in\unify_{\Eqsimp}(eqs)$ such that $\sigma(v_i)\neq\sigma(v_j)$ for any $\mathit{FrE}$-sorted variables $v_i$ and $v_j$, we set $\unify_{\Eqsimp}(eqs)=\{\sigma\}$.
\item Otherwise, we let $\unify_{\Eqsimp}(eqs)$ be the complete set of most general unifiers.
\end{itemize}

The rule that performs the above steps is as follows:
\begin{rul}[$C_{EqSimp}$]
\namedthmlabel{rul:CeqSimp} Assuming $\NoCanc(t)$,
$$\crule{\Large\substack{Eq(t, h),\ \Gamma}}{\{\Gamma' \}_{\sigma\in \{\unify_{\Eqsimp}(eqs) | eqs \in eqs(t, \roots(h))\}}}{C_{EqSimp}},$$
\end{rul}
\noindent where $\Gamma' = gen(\sigma)\Gamma \cup \{gen({\sigma})t \dheq gen(\sigma)h\}$.  
Recall that the substitution $gen(\sigma)$ generalizes the substitution $\sigma$ by extending every variable $v_i$ in the domain of the substitution from $v_i \mapsto term_i$ to $v_i\mapsto term_i + {Y_i}$, where $Y_i$ will later be interpreted as an unknown in linear equations. 

If $\NoCanc(t)$ instead does not hold, and unifiers for each of the root terms are found, Tamarin will add this as an assumption of the constraint system. 

In the cases where $t$ or $h$ are variables of compatible sorts, we can instead directly apply the following rule:
\begin{rul}[$C_{var}$]
\namedthmlabel{rul:var}
$$\crule{Eq(t,h, L), t \text{ variable}, t\notin Vars(h),\ \Gamma}{\sigma(\Gamma)}{C_{var}},$$
where $\sigma=\{t \mapsto h\}$.
\end{rul}
We also introduce a variant of this rule with the same name, with $t$ and $h$ interchanged. 
These rules state that when trying to unify a variable with a term, we can directly instantiate that variable with that term.

\subsection{Proofs of soundness and completeness}
\label{sec:proofsrules}

\subsubsection{Proof for rule \ref{rul:C@}} 

\begin{itemize}
\item \textbf{Soundness.} Soundness follows since the rule only adds constraints. Any model to the system with added constraints is also a model of the system without the additional constraints.

\item \textbf{Completeness. } Suppose there is a model before application of the rule. Then there exists a structure $(dg,\theta)$ that satisfies the constraint $Fact(t)@i$. By the satisfaction relations, this means that in the dependency graph $dg$ there is a node instance given by the rule $ru:i$ whose label is $Fact(h)$ such that $\theta(t) =_{\varepsilon_{DH}} \theta(h)$. Hence by definition the constraint $Eq(t, h)$ is also satisfiable. 
\end{itemize}

\subsubsection{Proof for rule \ref{rul:CPrem}} 
\begin{itemize}
\item \textbf{Soundness. } Since we are only adding more constraints, this rule is sound. 
\item \textbf{Completeness. } Suppose there is a model before application of the rule. Then there exists a structure $(dg,\theta)$ that satisfies the constraint system with the constraint $i:ri$. By the satisfaction relations, this means that in the dependency graph $dg$ there is a node instance given by the rule $i:ri$ whose premises are connected to other nodes. In particular there is an outgoing edge to a node $j:rj$ with a conclusion $Fact(h)$ such that $\theta(t) =_{\varepsilon_{DH}} \theta(h)$. Hence as in the previous proof, by definition the constraint $Eq(t,h)$ is also satisfiable. 
\end{itemize}

\subsubsection{Proof for rule \ref{rul:CeqSimp}. }

\begin{itemize}
\item \textbf{Soundness. } Let $(dg, \theta)$ be a solution of the constraint system after the application of the rule. Then $\theta \circ (gen(\sigma))$ is a solution of the system before the rule application. Indeed $Eq(\cdot)$ and $\dheq$ have the same semantics. 

\item \textbf{Completeness.}
Suppose that $(dg,\theta)$ is a solution of the constraint system before the
rule application. Thus,
\[
\theta(t) =_{DH} \theta(h)
\quad\text{and}\quad
\theta(\Gamma)\ \text{is satisfiable}.
\]
We show that there exists a unifier $\sigma$ produced by the rule and a
substitution $\tau$ such that
\[
\tau \circ gen(\sigma) = \theta,
\]
which implies that $(dg,\tau)$ is a solution after the rule application.

We distinguish two cases.

\paragraph{Case~1.}
Assume there exists a unifier $\sigma$ such that
\[
\sigma(v_i) \neq \sigma(v_j)
\quad\text{for all distinct variables } v_i,v_j \text{ of sort } \mathit{FrE}.
\]

For each variable $v$ of sort $\mathit{FrE}$, let $v_\sigma = \sigma(v)$.
Since $\theta$ must also map $\mathit{FrE}$ variables to $\mathit{FrE}$
variables, define $\tau$ on $\mathit{FrE}$-variables by
\[
\tau(v_\sigma) := \theta(v).
\]
This is well defined by the assumption on $\sigma$.

For variables $v$ of sort $E$, let $Y_v$ be the fresh variable introduced by
$gen$. Define
\[
\tau(Y_v) := \theta(v) - \sigma(v).
\]

Then by construction,
\[
\tau \circ gen(\sigma) = \theta,
\]
and therefore $(dg,\tau)$ satisfies $\Gamma'$.

\paragraph{Case~2.}
Assume that no unifier satisfies the condition of Case~1, and consider the full
set of most general unifiers returned by
$\unify_{\mathcal{E}_{\mathit{simp}}}$.

Let $t_1,\ldots,t_n$ be the root terms of $t$. By the linearity assumption,
no root term contains a product of $E$-variables.

We first consider the case where $t$ contains no $E$-variables.
Then $\theta(t_i)$ must be a root term for each $t_i$, since $\mathit{FrE}$
variables cannot be instantiated with sums.

Since $\theta(t) =_{DH} \theta(h)$, for each $t_i$ we have
\[
\theta(t_i)
=
\theta(h_{j_1})_{(k_1)}+\cdots +\theta(h_{j_s})_{(k_s)},
\]
where $\theta(h_{j_\ell})_{(k_\ell)}$ denotes the $k_\ell$-th root term of
$\theta(h_{j_\ell})$.

By $\NoCanc(t)$, and because $t$ does not contain $E$ variables, all but one of these factors must cancel, so there exists a
unique $\theta(h_{j})_{(k)}$ such that
\[
\theta(t_i) = \theta(h_j)_{(k)}.
\]

We now extract a unification query for each $t_i$:
\begin{itemize}
\item If $\theta(h_j)$ has exactly one root term, consider the equation
$t_i = h_j$, which is unified by $\theta$.
\item If $\theta(h_j)$ has multiple root terms, then (by linearity) $h_j$
contains exactly one $E$-variable $e$, and
\[
\theta(e) = f_1 + \cdots + f_m
\]
for irreducible monomials $f_\ell$. Then $\theta(h_j)_{(k)}$ contains one of these monomials, say $f_k$.
Define $\theta'$ to agree with $\theta$ except that
$\theta'(e) := f_k$, i.e. we ignore the other summands.
Then $\theta'$ unifies $t_i = h_j$.
\end{itemize}

Applying this construction to obtain unification queries for all $t_i$ yields a set of equations belonging to
$eqs(t,\roots(h))$, unified by some substitution $\sigma$. Since $\sigma$ is most general,
there exists a substitution $\tau$ such that $\tau \circ \sigma (t_i) = \theta(t_i)$ for root terms in the first case, or $\tau\circ\sigma(t_i) = \theta'(t_i)$ for root terms in the second case.

Recall that $gen(\sigma)$ maps each exponent variable $e$ to $\sigma(e) + Y_e$, where
$Y_e$ is a fresh variable. We therefore define a substitution $\rho = \tau' \circ \tau$,
where $\tau'$ maps each such $Y_e$ to $f_1 + \cdots + f_m - f_k$ in the second case.
This definition is well formed, since all variables introduced by $gen$ are fresh.

By construction, we obtain $\rho \circ gen(\sigma) = \theta$, and hence $(dg,\rho)$
satisfies the constraint system $\Gamma'$.

We now consider the case where a root term $t_i$ contains an $E$-variable $v$, then by linearity
\[
t_i = v \cdot m,
\]
where $m$ is an irreducible monomial containing only $\mathit{FrE}$ variables.
If $\theta(t_i)$ is a root term, the above argument applies directly.
Otherwise, $\theta(v)$ must be a sum, and by selecting a single summand we
again extract a unification query as above and define $\rho$ exactly with the same procedure. 
\end{itemize}

\subsubsection{Proof for rules \ref{rul:var}} 
\label{sec:basecases}
\begin{itemize}
\item \textbf{Soundness. } If $(dg, \rho)$ is a model of the system after applying this rule, then $(dg, \rho \circ \sigma )$ is a model of the original system. 
\item \textbf{Completeness. } Suppose $(dg, \rho)$ is a model of the system before the rule application. Then since $\rho(v)$ is equal to $\rho(X)$ and $v$ is a variable, we can split $\rho=\theta \cup \{v\mapsto \rho(X)\}$. Then $(dg, \theta)$ is a solution to our system after rule application. 
\end{itemize}

\subsubsection{Proof for rule \ref{rul:C=}} 
\label{sec:proofC=}
Recall the Rule $C_=$ defined in Section \ref{sec:C=}. 
\begin{rul}[$C_{=}$]
If Equation (\ref{eq:DHActionFacts}) obtained from the constraint $t \dheq h$ admits a solution, let $\{\sigma^j\}$ be the set of substitutions corresponding to a basis of the solution space.
$$\crule{t\dheq h,\ \Gamma}{ \{ \sigma^j(\Gamma) \}_j  }{C_{=}}.$$
\end{rul} 

\begin{itemize}
\item \textbf{Soundness}: We need to show that if there is a model of the constraint system after application of the rule, then there is a model \emph{before} application of the rule. 
Let $\sigma=\sigma^{j}$ be the substitution (the rule \ref{rul:C=} case splits over the different substitutions consisting of the basis of the solution space) for which the system after rule application admits a model $(dg , \theta)$.

Observe that $dom(\sigma) \subset \{Y_i\}$, where $\{Y_i\}$ are the freshly introduced variables of sort $E$. 
Define $\theta' = \theta \circ \sigma$ and $(dg, \theta')$ is a model of the system \emph{before} application of the rule. Indeed, $\theta'(\Gamma) = \theta(\sigma(\Gamma))$, hence all the constraint in $\Gamma$ are satisfied by $(dg, \theta')$. Similarly $\theta'(h)\dheq \theta'(t)$ is also satisfied, since $\sigma(t) = \sigma(h)$ (and hence also $\theta'(t) = \theta'(h)$), by definition of $\sigma$.

\item \textbf{Completeness}. We need to show that if there is a model $(dg,\theta)$ of the constraint system before application of the rule, then there is one also after. 
Recall that $(dg,\theta) \vDash t\dheq h$ if and only if $\theta(t) =_{DH} \theta(h).$ 

We split $\theta=\tau\theta'$, where $\theta'=\theta\restrict{\{Y_i\}}$. Then, $\theta'$ corresponds to a solution of the system $\tau(h)=\tau(t)$. 
In particular there exists a instantiation $\rho$ of the free $E$ variables in a particular $\sigma^j$ such that $\theta'=\rho\sigma^j$. 

Hence $(dg,\tau\rho)$ is a model of the constraint system after rule application for the case split corresponding to $\sigma^j$, since $\theta=\tau\theta'=\tau\rho\sigma^{j}$.
\end{itemize}

\subsection{Combining $\mathit{DH}$ and $\mathit{Msg}$ terms}
\label{app:combrule}

In this section we give the constraint solving rule that allows us to solve equality constraints containing terms of both sort $\mathit{DH}$ and of sort $\mathit{Msg}$. These equality constraints are introduced by Tamarin when unifying $Fact(A)$ with $Fact(B)$ via the rules $S_@$ or $S_{Prem}$. In the first case $Fact$ is an action fact, while in the second it is a premise fact.

\begin{rul}[$C_{comb}$]
\namedthmlabel{rul:Ccomb}
$$\crule{A=B, \ \Gamma}{\{\beta\circ \sigma\restrict{Msg}(\Gamma), \beta(\ex(\sigma\restrict{DH}))\}_{\sigma \in \unify_{\EqUsr}(\cl_1(A), \cl_1(B))}}{C_{comb}},$$
\end{rul}

\noindent where $\beta=\alpha_1 ... \alpha_n$ for each $\alpha_i \in \cl_2(A)\cup \cl_2(B)$ and we define the function $\ex$ below. 

Intuitively, the rule separates unification into a $\mathit{Msg}$ part and a
$\mathit{DH}$ part. The $\mathit{Msg}$ components are unified using
$\EqUsr$-unification on cleaned terms, while the $\mathit{DH}$ components are
reintroduced as explicit equality constraints to be solved using the dedicated
$\mathit{DH}$ rules from the previous sections.

\begin{definition}
The function $\ex$ transforms a substitution into a corresponding $\mathit{DH}$ equality constraint, i.e., $$\ex(\{\alpha\mapsto \beta \}):= Eq(\alpha, \beta). $$ 
When we have a substitution mapping multiple variables, we extend $\ex$ to multiple equality constraints. 
\end{definition}

In the following let $M$ denote all variables of sort $\mathit{Msg}$ but not $\mathit{DH}$. 
\begin{itemize}
\item \textbf{Soundnesss}: the rule does not introduce false solutions. 

Let $\sigma \in \unify_{\EqUsr}(\cl_1(A),\cl_1(B))$ be a unifier such that there exists a model $(dg, \theta)$ satisfying both
$\beta\circ\sigma\restrict{Msg}(\Gamma)$ and $\beta(\ex(\sigma\restrict{DH}))$.

Since $\theta$ satisfies $\ex(\sigma\restrict{DH})$, the substitution
$\theta\restrict{DH}$ equalizes the $\mathit{DH}$ subterms abstracted during
cleaning. Likewise, $\beta\sigma\restrict{Msg}$ equalizes the $\mathit{Msg}$ parts of
$A$ and $B$, and thus
$$ \theta \circ \beta \sigma\restrict{Msg}(A)=
\theta \circ\beta\sigma\restrict{Msg}(B).
$$
By assumption, $(dg, \theta \circ \sigma\restrict{Msg})$ also satisfies $\Gamma$.
Hence $(dg,\theta \circ \sigma\restrict{Msg})$ is a solution of the original constraint
system before application of the rule.

\item \textbf{Completeness. }We show that the rule does not remove any solutions. 
The idea is that any unifier of $A$ and $B$ is composed of unifiers of $cl(A)$ and $cl(B)$ since $\Sigma_{DH}$ and $\Sigma_{Usr}$ are disjoint and hence we can separate the substitutions of $Msg$ and $\mathit{DH}$ variables.

More precisely, let $(dg,\theta)$ be a model of the constraint system $\Gamma$ such that $\theta(A)=\theta(B)$. 
We decompose $\theta$ into $\theta\restrict{DH}$ and $\theta\restrict{M}$ according to the sort of variables. For any substitution in $\theta\restrict{M}$ that maps a $Msg$ variable to a $DH$ term $t$, we replace the term by a fresh variable $v$. This gives a new substitution $\theta'$ which is a unifier of $cl_1(A)$ and $cl_2(B)$ and the $v\mapsto t$ mapping corresponds to a subset of the $\beta$ function, i.e. $\beta\theta'=\theta\restrict{M}$.  By definition of $MGU$, there exists a $\phi$ such that: 
$$\phi\sigma=\theta'$$
Also, $\theta\restrict{DH}$ is a solution to $\beta(\ex(\sigma\restrict{DH}))$ since $\theta(A)=\theta(B)$. 

Finally, we obtain that $(dg,\phi\restrict{M}\circ \theta\restrict{DH})$ is a solution to the constraint system after rule application. Indeed 
 \begin{align*}
 & \phi\restrict{M} \circ \theta\restrict{DH} \circ \beta\circ\sigma\restrict{M}= \\
 & \theta\restrict{DH} \circ \phi\restrict{M} \beta\circ\sigma\restrict{M} \\
& \theta\restrict{DH} \circ \beta (\phi\restrict{M} \sigma\restrict{M}) \\
& \theta\restrict{DH} \circ \beta (\theta'\restrict{M}) =\\
& \theta\restrict{DH}\circ \theta\restrict{M} = \theta
\end{align*}
\end{itemize}

Finally, we can solve the constraints of the form $Eq(\alpha, \beta)$ by unifying indicators and producing $\dheq$ constraints as described in Section \ref{sec:newrules}.

\section{Optimization for $K(\cdot)$ premises.}
\label{sec:KINpremises}
We use a specialized rule to solve the premise $K(\cdot)$, representing the adversary's capabilities to perform $\mathit{DH}$ operations when solving algebraic equations, effectively replacing also Rule \ref{rul:C=}.


First, observe that for terms of the form $\mu(t)$, there are two options for the adversary to deduce this term. Either (i) the adversary deduces $t$ and applies the $\mu$ operator, or (ii) the adversary learns $\mu(t)$ directly.
We introduce a rule that performs this simple case split and hence preserves solutions.
\begin{rul}[$C_{\mu}$]
\namedthmlabel{rul:Cmu}
$$\crule{ i:ri,\ K(\mu(t))\in prems(ri), \Gamma}{
\Large\substack{ \{j: [K(t)]-> [K(\mu(t)], j<i, \Gamma\} \text{ or } \\
\{  j:ru, i:ri, j<i, t=_{DH}h,\ \Gamma \}_{\small\substack{ru\in \mathcal{P},\\ Out(\mu(h)) \in concs(ru)}} } }{C_{PremK}},$$
\end{rul}

We now consider the more general approach for general terms that are not $\mu$ terms.  Given such a term $t$ and constraint system, we first define all possible ways to categorize the exponents of $t$ that are compatible with the current leaked set.
\begin{definition}
Let $m_1,\ldots,m_n$ be terms and $L^{\Gamma}_i$ the leaked set corresponding to a constraint system $\Gamma$ at some timepoint $i$. We define $\mathit{VE}_i(m_1,\ldots,m_n)$ to be the set of all $E$ variables appearing in $m_1,\ldots,m_n$ that are \emph{not} in $L_i$:
$$\mathit{VE}_i(m_1,\ldots,m_n) = \{v\in \bigcup_j \mathit{vars(m_j)} \ |\ v:E, v\notin L_i \}$$
\end{definition}
\begin{definition}
Let $\Gamma$ be a given constraint system, let $i$ be a given timepoint with corresponding leaked set $L^{\Gamma}_i$, and let $m_1,\ldots,m_n$ be a set of terms.  We define 
$$\mathcal{L}^{\Gamma}(m_1,\ldots,m_n) = \{L_i\cup P \}_{P\in \mathcal{P}(\mathit{VE}_i(m_1,\ldots,m_n))}$$
where $\mathcal{P}(\cdot)$ denotes the power set operator. Hence $\mathcal{L}^{\Gamma}(m_1,\ldots,m_n)$ considers all possible ways of extending $L^{\Gamma}_i$ with elements from $\mathit{VE}_i(m_1,\ldots,m_n)$.
\end{definition}
Now let $t$ be a term such that $\roots(t) = \{t_1,\ldots,t_n\}$, and $t$ is not an $E$-variable. 
\begin{rul}[$C_{PremK}$]
\namedthmlabel{rul:CPremK} Assuming $\NoCanc(t)$,
$$\crule{ \Large{i:ri,\ K(t)\in prems(ri), \Gamma }}{\Large\substack{\{  i:ri,\ K(t)\in prems(ri),\ \Gamma_{L}\\  j_1:ru_1,\ ...,\ j_n:ru_n, \ j_k<i,\\  EqK(t, \langle h_1,\ldots, h_n\rangle, L) \} _{\substack{ru_1,\ldots,ru_n\in \mathcal{P},\\ Out(h_k) \in concs(ru_k),\\ L\in \mathcal{L}^{\Gamma}(t, h_1,\ldots,h_n)}}} }{C_{PremK}},$$
\end{rul}
\noindent where $\Gamma_{L}=\Gamma \cup \{K(e)@j_e, j_e<i\}_{e\in L\setminus{L^{\Gamma}_i}}$. The constraints $EqK(t, \langle h_1,\ldots, h_n\rangle, L)$ are defined in Section \ref{sec:semanticdefs}, and they pair each root term $t_k$ of $t$ with a term $h_k$ containing its indicator.  

We now define the respective rule for variables. 
\begin{rul}[$C_{varK}$]
\namedthmlabel{rul:CvarK} Let $v$ be an $E$-var,
$$\crule{ \Large{i:ri,\ K(v)\in prems(ri), \Gamma }}{\Large\substack{\{i:ri,\ K(v)\in prems(ri), \\ j:ru,\ j<i,\ EqK(v, h,L),\ \\  \ \Gamma_{L} \} _{ru \in \mathcal{P}, Out(h) \in concs(ru), L\in \mathcal{L}^{\Gamma}(h)} }}{C_{PremK}},$$
\end{rul}
Where again $\Gamma_{L}=\Gamma \cup \{K(e)@j_e, j_e<i\}_{e\in L\setminus{L^{\Gamma}_i}}$ and this time we don't consider $v\in L$ and we only case split over the exponents $e\neq v$ appearing in $h$.

The proofs of soundness and completeness for Rule~\ref{rul:CvarK} are identical to those for Rule~\ref{rul:CPremK}. We therefore present only a single proof.

\begin{itemize}
\item \textbf{Soundness.} The rule only adds constraints. 
\item \textbf{Completeness.} If there is a dependency graph satisfying the system before rule application, by Lemma \ref{lemma:Kchain} (see below), each root term $t_i$ of $t$ is connected by a $K$-chain to a rule containing an $Out(h_i)$ term such that $Ind(\theta(t_i))=Ind(\theta(h_j))$, and $h_j\in \roots(\theta(h_i))$. 
\end{itemize}

\begin{lemma}
\label{lemma:Kchain}
Let $t$ be a root term such that $\NoCanc(t_i,t_j)$ holds between its root terms. 
Suppose that a rule $ru$ has as its premise a $K(t)$ fact. Let $dg$ be an arbitrary dependency graph that has as its root the node containing an instance of this rule. Then for each root term $t_i$ of $t$ there is a path within $dg$ from $ru$ to a previous node $rj$ that has an $Out(\alpha)$ fact as conclusion of $rj$, and $\roots(\alpha)$ contains a term $\alpha_i$ such that $Ind(\alpha_i) = Ind(t_i)$. All other nodes of the path are instances of the rule $K(t_1),\ldots,K(t_n) --> K(f(t_1,\ldots,t_n))$.
\end{lemma}
\begin{proof}
The only rules that contain a $K(t)$ fact as conclusions where $t$ does not have a trivial indicator are the following
\begin{lstlisting}
recv: [Out(t)] -> [K(t)] 
f: [K(t1),...,K(tn)] -> [K(f(t1,...,tn))]
\end{lstlisting}
Now consider the subgraph of $dg$ obtained by starting from the root $ru$ and taking all nodes that are instances of the rule \texttt{f} and stopping when hitting a \texttt{recv} node instance. In particular this will create a tree where the leaves are \texttt{recv} node instances and all intermediary nodes are \texttt{f} instances. Let $T=\{\alpha_1,\ldots,\alpha_n\}$ be the terms of the leaves in the \texttt{recv} conclusions. Since they come from $Out$ facts, by the definition of our leaked set, the exponents in $T$ belong in $L$. Applying \texttt{f} rules corresponds to computing terms in $\Gen(T)$. Hence $t\in \Gen(T)$, so by Theorem \ref{thm:cdh}, any indicator of $t$ is also an indicator of one of the terms $\alpha_i$. 
\end{proof}

In a second step, when unifying these root/indicator terms in the simplified Diffie-Hellman theory, additionally to introducing variables directly within our terms, we also allow the adversary to manipulate and combine terms arbitrarily, representing that he can apply arbitrary DH operators. 

We mentioned that to do so, we will pass through algebraic systems. In particular, we will create equations over the leaked and the secret variables.

Again, suppose $\NoCanc(t)$ holds for a given term $t$, for which we want to solve $EqK(t, \langle h_1,\ldots, h_n\rangle, L)$, where $\mathcal{H}=\{h_1,...,h_n\}$ is a set of root terms coming from different $Out$ facts, and $L$ an exponent set.  Let $eqs_L(t, \mathcal{H})$ be the set of combinations of root elements of $h_1,\ldots,h_n$ and the associated set of equations obtained following the algorithm described for Rule \ref{rul:CeqSimp}, but instead of unifying $t_i$ with $h_i'$, we unify $Ind_L(t_i)$ with $Ind_L(h_i')$.

Consider a $\sigma\in\unify_{\varepsilon_{simp}}(eqs_L(t, \mathcal{H}))$. 
Similarly to Rule \ref{rul:CeqSimp}, we again consider a generalization of the substitution $\sigma$.

We will consider a function that differs slightly from the generalization process from Definition \ref{def:gen}. There, we generalized substitutions arising from the unification of entire root terms, whereas here we unify only their secret subcomponents. 
Furthermore, for the $h_i$ terms, which do not necessarily satisfy the $\NoCanc$ property, we must explicitly allow for monomials containing secret variables that may cancel each other out.
\begin{definition}
\label{def:genH}
Define $S_{\mathcal{H}} = \{ e_i\}_{i=1,...,p}$ to be the set of monomials in $\mathcal{H}=\{\sigma(h_1),...,\sigma(h_n)\}$ containing variables in the secret set $S$ that do not appear in $\sigma(t)$. 
\begin{align*}
&gen_{\mathcal{H}}(\sigma) (v) = \\
& \begin{cases}
\sigma(v) & \text{ if }v:\mathit{FrE} \\
gen(\sigma)(v) +\sum_{e_i\in S_{\mathcal{H}}} Y_i e_i & \text{ if }v:E
\end{cases}
\end{align*}
where $Y_1,...,Y_{|S_{h}|}$  are fresh $E$ variables.
\end{definition}

To take into account the adversary capabilities, we must check if the $gen_{\mathcal{H}}(h_i)$ can actually be combined, possibly with other adversarially known terms to construct $gen_{\mathcal{H}}(t)$. 

More precisely, we further introduce new variables $X_k$ for each term $h_k$ that represent the adversary exponentiating $h_k$ with some terms. We also introduce an extra variable $X_{n+1}$ representing the adversary multiplying the $h_i$'s with a term of his knowledge. We hence also introduce the corresponding constraints $K(X_l)@j_l$ with $j_l<i$ and $K(X_{k+1})@j_{k+1}$ with $j_{k+1}<i$, where $i$ is the timepoint of the rule whose premise $K(t)$ we are solving. 

Again, as the $X_i$ are only solved in $\mathbb{Q}(L)$, we must further explicitly allow the adversary to multiply with the other secret monomials in $S_{\mathcal{H}}$.
We introduce $p$ more variables $W_p$, add the constraints $K(g^{\sum_{e_i\in S_{\mathcal{H}} }W_ie_i})@j_w$ and $j_w<i$, and solve the following system of equations: 

\begin{align*}
& \operatorname{eq}_{t,h_1,..,h_n} := \\
& \quad gen_{\mathcal{H}}(\sigma)(t) \dheq \\
& \quad \sum_{i=1}^{n}X_igen_{\mathcal{H}}(\sigma)(h_i) +X_{n+1} + \sum_{e_i\in S_{\mathcal{H}} }W_igen_{\mathcal{H}}(e_i) . 
\end{align*}

Observe that the equation above is not necessarily a linear equation, since it may involve products of variables of the form $X*Y$.  
Since the current implementation supports only linear equation-solving algorithms, we can address this limitation by replacing each product $X*Y$ with a fresh variable $Z$. We then solve the system of equations using the Gauss elimination algorithm, obtaining a solution $Z=t$. Provided that the resulting solution term does not contain the $+$ function symbol, we subsequently apply variant-based unification to compute a complete set of solutions to the equation $X*Y=t$. Supporting the more general case in which this assumption is not satisfied would require implementing more advanced algebraic techniques, such as those based on Gröbner bases (e.g.\ Buchberger's algorithm). However, such cases have not arisen in our examples. 

\begin{rul}[$C_{EqSimpK}$]
\namedthmlabel{rul:CeqSimpK} Assuming $\NoCanc(t)$,
$$\crule{\Large\substack{
Eq_K(t, \langle h_1,\ldots,h_n\rangle, L),\ \Gamma}}{\{\Gamma_{\sigma} \}_{\sigma\in \unify_{\Eqsimp}( eqs_L(t,\mathcal{H}) )}}{C_{EqSimp}},$$
\end{rul}
\noindent where $\{ i:ri,\ K(t)\in prems(ri), j_1:ru_1,\ldots, j_n:ru_n, j<i,\}\subset \Gamma$ and 
\begin{align*}
\Gamma_{\sigma} =& \sigma(\Gamma) \cup \{K(g^{X_{n+1}})@j_{n+1}\} \cup \{K(X_i)@j_i, j_i<i\}_{i=1}^{n} \\
& \cup  \{K(g^{\sum W_igen_{\mathcal{H}}(\sigma)(e_i)})@l_i, l_i<i \}_{i=1}^{p}  \cup\\
& \operatorname{eq}_{t,h_1,..,h_n}. 
\end{align*}
Again, we distinguish two cases for $\unify_{\Eqsimp}$. 
\begin{itemize}
\item If Tamarin returns a substitution $\sigma\in\unify_{\Eqsimp}(eqs_L(t,\mathcal{H}))$ such that $\sigma(v_i)\neq\sigma(v_j)$ for any $\mathit{FrE}$-sorted variables $v_i$ and $v_j$, we set $\unify_{\Eqsimp}(eqs_L(t,\mathcal{H}))=\{\sigma\}$. 
\item Otherwise, we let $\unify_{\Eqsimp}(eqs_L(t,\mathcal{H}))$ be the complete set of most general unifiers.
\end{itemize}

The proof that this rule preserves solutions of the system and doesn't introduce spurious one is similar to the proof concerning Rule \ref{rul:CeqSimp}. 
\begin{itemize}
\item \textbf{Soundness} Let $(dg,\theta)$ be a model of the constraint system after rule application, $(dg,\theta\sigma)$ is a model of the constraint system before rule application. 

\item \textbf{Completeness} Let $(dg,\theta)$ be a model of the constraint system before rule application. Analogously to the proof concerning Rule \ref{rul:CeqSimp}, we can show there exists $\tau$ such that $\tau\sigma = \theta$ on the variables that are not the $W_i$ or $X_i$ variables. Since $(dg, \theta)$ satisfied the constraints $K(t)$ and $Out(h_i)\in concs(ru_j)$, in the dependency graph $dg$, there is a construction of $\theta(t)$ from the $\theta(h_i)$. 

If in this construction the adversary uses the $exp$ operator, computing $h_i^{e_K}$ for some exponent $e_K$ (that must be of his knowledge, since the rule in $dg$ requires the fact $K(e_K)$ as premise, we set $\tau_X:X_i\mapsto e_K$, and $K(\tau_K(X_i))$ is satisfiable, since $e_K$ must previously by known by the adversary in $(dg,\theta)$. 

If the adversary uses the multiplication operator, there are two cases:

\begin{itemize}
\item The adversary multiplying the terms $h_i$ with some other terms $g_i$ that only contain exponents that are known by the adversary. Then $g_i=g^f_i\in Gen(L)$, we set $\tau_X:X_{k+1} = \sum_i f_i$, and again $K(\tau_X(X_{k+1}))$ is satisfiable, since $K(\sum_f{f_i})$ is also solved in $(dg,\theta)$.

\item If the adversary uses the multiplication operator multiplying the terms $\theta(h_i)$ with some terms $g_i$ that also contain secret exponents, they must be terms used to cancel out other root terms of $h_i$. Indeed, recall that we assume $\NoCanc(t)$ holds and all the indicators of $t$ are already found in the $h_i$. So the secret exponents in the $g_i$ terms that do not appear in $t$ must appear in the $h_i$ terms.  Hence, $g_i$ is of the form $g^{l_is_i}$, where $s_i = \theta(e_i)$ for some $e_i$. 
In this case we can set $\tau_W:W_i\mapsto l_i$, and $K(\tau_W(g^{\sum W_i \sigma(e_i)}))@j$ will also be satisfiable by $(dg, \tau)$ since $(dg,\theta)$ also satisfies the constraint $K(g^{\sum W_i e_i})@j$. 
 \end{itemize}
Hence, we have shown that $(dg, \tau_W\tau_X\tau)$ is a model to the constraint system after rule application. 
\end{itemize}

As in Section \ref{sec:C=}, by abuse of notation, we will assume terms $t$ and the $h_i$ to be of the form $g^t$ and $g^{h_i}$.  
We now proceed to find a solution to the following equation with coefficients in $\mathbb{Q}(l_1,...l_r)$ for the given variables $\{W_i\}\cup \{X_i\}\cup \{Y_i\}$
\begin{align*}
& \operatorname{eq}_{t,h_1,..,h_n} := \\
& \quad gen_{\mathcal{H}}(\sigma)(t) \dheq \\
& \quad \sum_{i=1}^{n}X_igen_{\mathcal{H}}(\sigma)(h_i) +X_{n+1} + \sum_{e_i\in S_{\mathcal{H}} }W_igen_{\mathcal{H}}(e_i). 
\end{align*}

Matching each coefficient of the polynomials, this gives a system of equations in the variables $\{W_i\}\cup \{X_i\}\cup \{Y_i\}$.

\begin{rul}[$C_{=K}$]
\namedthmlabel{rul:C=K}
If Equation $\operatorname{eq}_{t,h_1,..,h_n}$ admits a solution, let $\{\sigma^j\}_j$ be the corresponding solution substitutions. 
$$\crule{{\Large \substack{\operatorname{eq}_{t,h_1,..,h_n},\ K(g^{\sum(W_iE_i)})@j, \\ \ K(g^{X_{k+1}})@j_{k+1},\ K(X_i)@j_i},\ \Gamma} }{ \{ \sigma^j(\Gamma') \}  }{C_{=}}.$$
where $\Gamma' = \Gamma \cup \{K(g^{\sum(W_iE_i)})@j,\ K(g^{X_{k+1}})@j_{k+1}\}\cup \{K(X_i)@j_i\}$
\end{rul}

\begin{itemize}
\item \textbf{Soundness.} We first show that this rule is sound.
Let $(dg,\theta)$ be a solution to the constraint system $\sigma(\Gamma)$. To prove soundness, we use an argument analogous to that in Section \ref{sec:proofC=}. 
If we define $\theta' = \theta \circ \sigma$, then $(dg, \theta')$ is a model of the system \emph{before} application of the rule. Indeed, $\theta'(\Gamma) = \theta(\sigma(\Gamma))$, hence all the constraints in $\Gamma$ are satisfied by $(dg, \theta')$. 

\item \textbf{Completeness.} As for completeness, suppose we have a solution $(dg,\theta)$ of the constraint system before rule application. Because the $K(W_i)$ and $K(X_i)$ premises are satisfied by $dg$, $\theta(X_i), \theta(W_i)\in Gen(\theta(L))$. Similarly, because $\sigma$ is the most general unifier of the secret variables, the remaining $\theta(Y_i)\in Gen(\theta(L))$.
Then $\theta$ must be a solution of the equation system $\operatorname{eq}_{t,h_1,..,h_n}$, and there is an instantiation $\rho$ of the $W_i$, $X_i$ and $Y_i$ variables such that $\rho \sigma^{j} = \theta\restrict{\{W_i, X_i, Y_i\}}$. Hence $(dg,\theta\restrict{\mathcal{V}\setminus \{W_i, X_i, Y_i\}}\circ \rho)$ is a solution to the constraint system corresponding to case $\sigma^{(j)}$. 
\end{itemize}

\section{Tamarin model of MQV}
\label{sec:mqvmodel}

We show our Tamarin model of MQV and explain the reasoning steps the extended Tamarin prover tool is now following. Note that our newly added builtin of extended Diffie-Hellman is called \texttt{DH-multiplication}.

\begin{lstlisting}[]
theory MQV
begin
builtins: DH-multiplication, symmetric-encryption

rule GenKey: 
   [ Fr(sk:FrE) ] 
   --[  ]->
   [ !SKey($A, sk:FrE), !PubKey($A, g^sk:FrE),
     Out(g^sk:FrE) ]

rule GenKeyCompromised:
    [K(sk:E)]
   --[ Compromised($A) ]->
   [ !PubKey($A, dhExp(g, sk:E) ) ,
     Out(g^ sk:E) ]

rule InitiatorRole:
  [!SKey($I, a:FrE),!PubKey($R, g^b:E)), 
   Fr(x:FrE)]
  --[ Neq($I, $R) ]->
  [Out(g^x:FrE), 
   Initiated($I, $R, a:FrE, x:FrE, g^b:E)]

rule ReceiverRole:
let kAB = (g^x:E . g^(a:E* mu(g^x:E)))^
                  (y:FrE + b:FrE*mu(g^y:FrE))
in
  [!SKey($R, b:FrE), !PubKey($I, g^a:E), 
  In(g^x:E), Fr(y:FrE), Fr(m:FrE)]
  --[ Neq($I, $R), Neq(a:E, b:FrE), 
      Neq(x:E, y:FrE), RunningR(R,I, kAB)]->
  [ Out(g^y:FrE), Out(senc(g^m:FrE, kAB)), 
    ReceiverSend($R,$I,kAB,m:FrE) ]

rule InitiatorRole2:
let kAB = (g^y:E . g^(b:E*mu(g^y:E)))^ 
	      (x:FrE + a:FrE*mu(g^x:FrE))
in
  [ In(g^y:E)), Fr(m:FrE), In(senc(g^mr:E, kAB))
  Initiated($I, $R, a:FrE, x:FrE, g^b:E), ]
  --[Neq(y:E, x:FrE), Neq(a:FrE, b:E), 
     AgreeKeyI($I, $R, kAB)]->
  [Out(senc(g^m:FrE, kAB)) ]

rule ReceiverRole2:
  [ In(senc(g^mi:E, kAB)), 
    ReceiverSend($R, $I, kAB, m:FrE)]
  --[Neq(m:FrE,mi:E), AgreeKeyR($R, $I, kAB)]->
  [ ]
  
  end
  \end{lstlisting}

\subsection{Agreement property}
We first explain how Tamarin disproves the property
\begin{lstlisting}
lemma agreementI:
"All I R K #i. AgreeKeyI(I, R, K)@i & (not (Ex #k. Compromised(I)@k))& 
(not (Ex #k. Compromised(R)@k)) 
 ==> Ex #j. Running(R, I, K)@j"
 \end{lstlisting}
In particular, this will reconstruct the attack by Kaliski. 
Tamarin searches for a counterexample, in a backward manner, starting from the action fact $AgreeKeyI(I,R,K)$ trying to build a trace that does \emph{not} contain the action fact $RunningR(R,I,K)$.

The action fact $AgreeKeyI$ only appears in the rule \texttt{InitiatorRole2}, implying that we must have:

\begin{align*}
K & = (g^{y_E}\cdot g^{b\cdot \mu(g^{y_E})})^{x+\mu(g^x)a}\\
& = g^{y_Ex}\cdot g^{b\mu(g^{y_E})x}\cdot g^{y_E\mu(g^x)a} \cdot g^{b\mu(g^{y_E})\mu(g^x)a},
\end{align*}
where $g^b$ is Bob's (authentic) public key and $g^{y_E}$ is the ephemeral key received. 

 Since Alice checks the validity of the encrypted message, the received message must be encrypted with that key. Hence the message either comes from the rule \texttt{ReceiverRole} using the same key, or the adversary knows the key and has encrypted the message himself (Tamarin proved the key secret in the lemma \texttt{secrecyI} from Section \ref{sec:exmqv}, so can easily exclude the latter case). The key in \texttt{ReceiverRole} is
 \begin{align*}
 K_B & = (g^{x_E}\cdot g^{e_E \mu(g^{x_E})})^{y+\mu(g^{y})b} \\
 & = g^{x_E y}\cdot g^{e_E \mu(g^{x_E})y} \cdot g^{x_E\mu(g^{y})b}\cdot g^{e_E\mu(g^{x_E})\mu(g^{y})b},
 \end{align*}
\noindent where $g^{x_E}$ is the ephemeral key received, and $g^{e_E}$ is the (possibly malicious) public key of the agent $E$ Bob is talking to. In particular, note that we cannot have that $E=I$, otherwise this trace would not be a counterexample (recall that our trace cannot contain the action fact $RunningR(R,I,K)$). 

As mentioned, since Alice checks these two keys match, Tamarin inserts the constraint

$$Eq(K, K_B).$$

The root terms of $K$ are $$g^{x},\ g^{b\mu(g^{y_E})x},\ g^{y_E\mu(g^x)a} ,\ g^{b\mu(g^{y_E})\mu(g^x)a}$$

We then consider all permutations of root terms of $K_B$ and unify root terms. We consider here only a permutation that leads to an attack. In particular, we try to unify the
root terms with only the subterms $g^{x_Ey}$ and  $g^{x_E\mu(g^y)b}$. 
The only variable of sort $E$ in these two terms is $x_E$, so we replace it by $x_E = f_1+f_2$, obtaining the four root terms $$g^{f_1*y}, g^{f_2*y}, g^{f_1*\mu(g^y)*b}, g^{f_2*\mu(g^y)*b}.$$ Consider the permutation corresponding to the unification queries

\begin{align*}
& g^{x*{y_E}} = _{\varepsilon_{simp}} g^{f_1 *y} ,\\
& g^{b*\mu(g^{y_E})*x} = _{\varepsilon_{simp}} g^{f_1*\mu(g^y)*b}\\
& g^{y*\mu(g^x)*a} = _{\varepsilon_{simp}} g^{f_2*y} \\
& g^{b*\mu(g^{y_E})*\mu(g^x)*a} = _{\varepsilon_{simp}} g^{f_2*\mu(g^y)*b}
\end{align*}
Solving these equations gives a possible unifier: $\sigma=\{f_1\mapsto x, f_2\mapsto a\mu(g^x), y_E \mapsto y\}$. Its generalization is given by:
$$\{x_E\mapsto x + a\mu(g^x)+ Y_1, e_E \mapsto Y_2, y_E\mapsto y + Y_3\}.$$

Recall that we abstract $\mu$-terms away, so we introduce new variables $\mu_y,\mu_{y_E},\mu_x$ and $\mu_e$, keeping track of the substitutions:
\begin{align*}
&\mu_y=\mu(g^y),\\
&\mu_{y_E}=\mu(g^{y+Y_3}),\\
&\mu_x = \mu(g^x), \\
&\mu_e = \mu(g^{x + a\mu(g^x) + Y_1}).
\end{align*}

We then consider all possible equalities of $\mu$ variables. In particular, we consider the case where $\mu_y = \mu_{y_E}$ and the others are different, which corresponds to the substitution $Y_3=0$. 

Finally, we get the equation:
\begin{equation}
\label{eq:ex}
gen(\sigma)(K_B) = gen(\sigma)(K).
\end{equation}
which, when instantiating $Y_3=0$, corresponds to the long equation of exponents:
\begin{align*}
&(y + b\mu_y)(x+\mu_x a)=\\
&(x+a\mu_x+Y_1+Y_2\mu_e)(y+\mu_yb).
\end{align*}

Observe that $Y_1$ appears inside the term that the variable $\mu_e$ replaces. Hence if we were to solve the equation for $Y_1$, it would contain the variable itself.   
Hence we instead solve the equation for $Y_2$, yielding $$Y_2 =-\frac{Y_1}{\mu_e}.$$
We have one free variable $Y_1$, so we set $Y_1=v_{Y}$, for a new variable $v_{Y}$ of sort $\mathit{E}$.

We plug this back in our constraint system, and we need to solve for the premises of rule \texttt{ReceiverRole}. 
In particular, we need to solve $$In(g^{x+\mu_xa +v_{Y}}),$$
$$!PubKey(\$E, -\frac{v_Y}{\mu_e}).$$

Consider first the constraint $In(g^{x+\mu_xa +v_Y})$. Using the $C_{var}$ on the rule \texttt{K(x) -> In(x)}, we know this must come from the fact: $$K(g^{x+\mu_xa +v_Y}).$$ The term is a non-variable term of sort $G$, so we apply the Rule \ref{rul:CPremK} to try to construct this term from possible $Out$ terms of sort $G$. In particular, there are three rules that have as conclusion an $Out$ fact: $Out(g^x)$ (Alice's key generation rule), $Out(g^a)$ (rule \texttt{InitiatorRole}), and $Out(g^{f})$ for a fresh adversarial $\mathit{FrE}$ variable $f$. 

So far, there are no rules in our constraint system $\Gamma$ that contain an action $K(e)$ for any exponent variable $e$. This means for any timepoint $i$ our current leaked set is $L^{\Gamma}_i = \emptyset$.  We hence need to consider all possible cases $L\in \mathcal{P}(\{x, \mu_x, a, v_Y,f\})$. For each of the above $E$ variables $e\in \{x, \mu_x, a, v_Y,f\}$ we consider whether $e\in L$ and in that case add the constraint $K(e)@j_e$. For all of the fresh variables we first try to solve these goals. 

Tamarin proves  $K(x)$ and $K(a)$ are not solvable, since both $x$ and $a$ are fresh and hence they are not unifiable (in the simplified $DH$ equational theory) with any root term of an $E$-argument appearing in $Out$ facts. This excludes the cases where $x\in L$ and $a\in L$. 

Tamarin finds a solution for $K(\mu_x)@j_{\mu_x}$. Recall that $\mu_x= \mu(g^x)$. Indeed, Tamarin first applies the rule that case splits whether the adversary learns the entire term, or first the argument and then applies the $\mu$ operator. 
We consider the latter case. This introduces the constraint $K(g^x)@k, k<j_{\mu_x}$. This is solvable by unifying the indicator of $g^x$ (which is again $g^x$, since Tamarin has already shown $x\in S$) with the $Out(g^x)$ conclusion coming from Alice's key generation rule. The equation generated is $X_1x +X_2 = x$, which, since $x$ is secret, has a unique solution $X_1=1$ and $X_2=0$. The constraints $K(1)$ and $K(0)$ are solved using the \texttt{zero} and \texttt{one} rules.

Tamarin also finds a solution for the $K(f)$ constraint (since the rule \texttt{FrNZE} also has an $Out(f)$ fact). Hence for now Tamarin has shown $\mu_x\in L, f \in L, x\notin L, a\notin L$.
 
 The variable $v_Y$ is of sort $E$ so we delay solving the constraint $K(v_Y)$ later. We hence only have 2 cases to consider: $L=\{\mu_x,f, v_Y\}$ or $L=\{\mu_x,f\}$, where the first case also introduces the  $K(v_Y)@j_{v_Y}$ constraint in the system. 

Recall that we have three rules whose conclusions contain $Out$ facts with arguments: $g^x$, $g^a$ and $g^f$.  For both possible $L$'s, we must consider all possibilities of matching the indicators of the term $g^{x+\mu_xa +v_Y}$ with the indicators of these $Out$ terms. We consider the case where $L=\{\mu_x,f, v_Y\}$. In this case $Ind_L(g^{v_Y}) = e_g$, so we don't need to further search for this indicator. Consider the case where we try to unify the remaining indicators $Ind_L(g^x)$ and $Ind_L(g^{\mu_xa})$ with the indicators of the terms $g^x$ and $g^a$. 

\begin{align*}
& g^{x}= g^x, \\
& g^{a} = g^a
\end{align*}
These terms are already equal, so we just generalize the $E$ variables in our system that are not modified by the substitution, in this case $v_Y\mapsto Y$.
There are no secret variables appearing in the RHS terms $g^x$ and $g^a$ that are not already in the LHS terms $g^x, g^a$, and $g^{v_Y}$. We hence only introduce new variables $X_1$, $X_2$, and $X_3$ and we get the equation:
$$x+\mu_xa + Y = xX_1 + aX_2 + X_3.$$

We add the new constraints $K(X_i)@j_i$ for $i=1,\ldots,3$. 
Since $x, a\in S$, this corresponds to the equations
\begin{align*}
& X_1 = 1, \\
& X_2 = \mu_x, \\
& X_3 = Y.
\end{align*}
The solution is $X_3=Y$, where we replace $Y$ again with a free $E$-variable $v_Y$, and the introduced constraints we need to solve become: 
$$K(1)@j_1, K(\mu_x)@j_2, \text{ and } K(v_Y)@j_3.$$
$K(1)@j_1$ is solved using the \texttt{one} rule. 
We always solve goals of the form $K(v)@j$, with $v$ a variable of sort $E$, such as $K(v_Y)@j_3$, last. 

The constraint $K(\mu_x)@j_2$ is solved exactly like the $K(\mu_x)@j_{\mu_x}$ constraint solved above.  

Finally, let us solve the $!PubKey$ action fact. This cannot come from the honest key generation rule, as $-\frac{v_Y}{\mu_e}$ can never be matched by a fresh variable. 
It must hence come from the dishonest key generation. For that, we must solve that rule's premise
$$K(-\frac{v_Y}{\mu_e}).$$
The indicator of $-\frac{v_Y}{\mu_e}$ is $\mu_e$, and again this can be constructed by first constructing its argument $g^{x+\mu_xa +v_Y}$. Hence this goal is solved in the same manner as the previous one. 

At this point, the only constraint we have left is the constraint $K(v_Y)@j_3$ we have postponed before. By solving this goal, we solve the constraint system. Observe that the choice $v_Y=f$, obtained by solving the goal $K(v_Y)$ with the \texttt{FrNZE} rule, corresponds to the attack described by Kaliski. The choice $v_Y=0$, obtained by solving the goal $K(v_Y)$ with the \texttt{zero} rule, corresponds to the attack using a trivial public key described in the main body of the paper.

\subsection{Secrecy property}

We next illustrate how Tamarin proves the lemma
\begin{lstlisting}
lemma secrecyI:
 "All I R key #i. AgreeKeyI(I, R, key)@i & (not (Ex #k. Compromised(I)@k))& 
 (not (Ex #k. Compromised(R)@k))
  ==> not (Ex #i. K(key)@i)"
\end{lstlisting}
Again, the action fact $AgreeKeyI$ only appears in the rule
\texttt{InitiatorRole2}, hence $C_{\mathit{var}}$ instantiates \texttt{key} with 
 \begin{align*}
 key & = (g^{y}\cdot g^{b \mu(g^{y})})^{x+\mu(g^{x})a} \\
 & = g^{yx}\cdot g^{b\mu(g^{y})x} \cdot g^{y\mu(g^{x})a}\cdot g^{ab\mu(g^{y})\mu(g^{x})},
 \end{align*}

where $g^b$ is Bob’s public key and $g^y$ is the
ephemeral key supposedly sent by Bob (but adversarially controlled). 

Once we have added the \textit{i:InitiatorRole2} to our contraint system, Tamarin will immediately try to solve the premises of that rule (Tamarin's heuristics always prioritize protocol state premise goals). In particular we solve the \textit{Initiated} premise which in turn appears only in the \texttt{InitiatorRole} rule. Again, Tamarin immediately solves the premises of \texttt{InitiatorRole}, namely we choose the \textit{PubKey} premise. 
  
This time, there are two rules that generate a \textit{PubKey} fact: \texttt{GenKey} and \texttt{GenKeyCompromised}. Tamarin excludes the compromised case, as the lemma statement excludes \textit{Compromised(B)} from appearing in the trace. Hence $b$ is instantiated with a fresh variable $sk_B$.

Our key variable thus becomes:
 \begin{align*}
 key & = (g^{y}\cdot g^{sk_B \mu(g^{y})})^{x+\mu(g^{x})a} \\
 & = g^{yx}\cdot g^{sk_B\mu(g^{y})x} \cdot g^{y\mu(g^{x})a}\cdot g^{a*skB\mu(g^{y})\mu(g^{x})},
 \end{align*}

The only rule that contains an action fact $K(t)@i$, for $t$ a $DH$ term is the rule \texttt{K(x) -[K(x)]--> In(x)}.
Using the $C_{var}$, we add this rule to the system and need to solve its premise: $$K(g^{yx}\cdot g^{sk_B\mu(g^{y})x} \cdot g^{y\mu(g^{x})a}\cdot g^{a*sk_B\mu(g^{y})\mu(g^{x})}).$$ The term is a non-variable term of sort $G$, so we apply the Rule \ref{rul:CPremK}. Observe that at this point $\NoCanc(key)$ does \emph{not} hold, as $\mathit{y:E}$ is adversarial and a priori could be chosen to cancel out some of the root terms (consider for example $y \mapsto \dhInv(x)$). 
  
Tamarin tries to apply $\NoCanc$ anyway, keeping track that the only weaker properties that hold are 
\begin{align*}
& \NoCanc(g^{sk_B\mu(g^{y})x}, key) \text{ and,} \\
& \NoCanc(g^{sk_Ba\mu(g^{y})\mu(g^{x})}, key).
\end{align*}
At this point Rule \ref{rul:CPremK} determines all possible leaked sets. The exponents appearing in the term $key$ are: $\{x,y,a,sk_B, \mu(g^y), \mu(g^x)\}$.
As in previous example, Tamarin shows that $\mu(g^y)\in L$ and  $\mu(g^x)\in L$. It also shows that $sk_B\notin L$ and $a\notin L$.

Hence the two cases to consider are $L=\{\mu(g^y), \mu(g^x)\}$ and $L=\{\mu(g^y), \mu(g^x),y\}$, where in the second case we add the constraint $K(y)@j$, with $j<i$.

For both cases:
\begin{align*}
& Ind_L(g^{sk_B\mu(g^{y})x})= g^{sk_B*x} \\
& Ind_L(g^{sk_Ba\mu(g^{y})\mu(g^{x})}) = g^{sk_B*a}
\end{align*}

Tamarin will try to unify these indicators with all possible root terms of $Out$ facts from the protocol rules, in the simplified $DH$-equational theory. However, none of these unification queries will succeed, as neither $g^{sk_B*x}$ nor $g^{sk_B*a}$ is contained in any root term of any message, not even as argument of an encrypted message of which the adversary knows the key. 

This implies the adversary cannot learn these indicators (even just one would have been enough). Since these terms are not cancellable, this shows that the adversary cannot learn the whole term, and Tamarin marks this lemma as solved. 

\end{document}